\documentclass[12pt]{article}
\usepackage[section]{placeins}
\usepackage{amsmath,amsfonts, latexsym, amssymb, amsthm}
\usepackage{natbib, url, setspace, lscape, graphicx}
\usepackage[top=1in, bottom=1in, left=1in, right=1in]{geometry}
\usepackage{tikz, ctable, soul, bbm, xcolor, lscape, pdflscape,dcolumn}

\newtheorem{theorem}{Theorem}
\newtheorem{lemma}{Lemma}
\newtheorem{proposition}{Proposition}
\newtheorem{corollary}{Corollary}

\DeclareMathOperator{\var}{var}

\DeclareMathOperator{\supp}{supp}

\begin{document}

\title{ \textbf{\Large Quantile Treatment Effects in Difference in
Differences Models under Dependence Restrictions and with only Two Time
Periods}\thanks{
We thank an Editor and two anonymous referees for their
constructive comments which have greatly improved the paper. We also thank
Yi-Ting Chen, Le-Yu Chen, Jiti Gao, Wolfgang Hardle, Yu-Chin Hsu, Hidehiko
Ichimura, Brett Inder, Kengo Kato, Jen-Che Liao, Daisuke Nagakura,
Anastasios Panagiotelis and participants of seminars at Academia Sinica,
Keio University, Tokyo University and Monash University. The last author
gratefully acknowledges the financial support from Singapore Ministry of
Education Academic Research Fund Tier 1 (FY2015-FRC3-003). } }
\author{ Brantly Callaway\thanks{
Department of Economics, Temple University (brantly.callaway@temple.edu) } 
\hspace{1.2cm} Tong Li\thanks{
Department of Economics, Vanderbilt University (tong.li@vanderbilt.edu) } 
\hspace{1.2cm} Tatsushi Oka\thanks{
Department of Economics, National University of Singapore (oka@nus.edu.sg) } 
\vspace{0.5cm} }
\date{ The first version: December 15, 2015 \\
This version: \today   
       \vspace{0.0cm} }
\maketitle

\begin{abstract}
This paper shows that the Conditional Quantile Treatment Effect on the
Treated can be identified using a combination of (i) a conditional
Distributional Difference in Differences assumption and (ii) an assumption
on the conditional dependence between the change in untreated potential
outcomes and the initial level of untreated potential outcomes for the
treated group. The second assumption recovers the unknown dependence from
the observed dependence for the untreated group. We also consider estimation
and inference in the case where all of the covariates are discrete. We
propose a uniform inference procedure based on the exchangeable
bootstrap and show its validity. We conclude the paper by estimating the effect of
state-level changes in the minimum wage on the distribution of earnings for
subgroups defined by race, gender, and education.
\end{abstract}

\vspace{0.5cm}

\vspace{-0.5cm}

\setcounter{page}{0} \thispagestyle{empty}

\vspace{1cm}


\ 


\newpage

\section{Introduction}

Researchers and policy makers are interested in evaluating the effect of
participating in a program or experiencing a treatment but this is not a
trivial task due to self-selection. \cite{Gronau1974JPE} and \cite%
{Heckman1974Emtca} study the issue of self-selection in the context of labor
markets and \cite{Amemiya1985book} provides a comprehensive framework for
this issue in his influential book \textit{Advanced Econometrics}, which
also includes his seminal paper \cite{Amemiya1973Emtca} on the Tobit model.
In the literature, the Average Treatment Effect (ATE) or Average Treatment
Effect on the Treated (ATT) has received great attention. But there are
cases where a researcher may be interested in studying the distributional
effect of treatment. In fact, estimating distributional treatment effect
parameters is becoming more common in applied work. Recently, distributional
treatment effects have been estimated in the context of welfare reform %
\citep{bitler-gelbach-hoynes-2006, bitler-gelbach-hoynes-2008}, conditional
cash transfer programs in developing countries \citep{djebbari-smith-2008},
head start \citep{bitler-hoynes-domina-2014}, and the effect of Job Corps %
\citep{eren-ozbeklik-2014}. One thing that each of the above empirical
papers have in common is that each uses experimental data.  In practice, however, experimental data is not always available and research for estimating
distributional treatment effects in observational studies is of importance.


The current paper considers identification and estimation of a particular
distributional treatment effect parameter called the Conditional Quantile
Treatment Effect on the Treated (CQTT) under a Difference in Differences
(DID) setup when only two periods of panel or repeated cross sections data
are available. For applied researchers, it is not uncommon to have exactly
two periods of data. For example, 25\% of the papers employing DID
assumptions considered by \cite{bertrand-duflo-mullainathan-2004} used
exactly two periods of data. To give some specific examples, the Current
Population Survery (CPS) Merged Outgoing Rotation Groups contains a 2-period
panel \citep[see][for instance]{madrian-lefgren-2000, riddell-song-2011} and
the Displaced Workers Survey contains data on current wages and wages before
displacement \citep[see][for example]{farber-1997}.

For identifying the CQTT, we consider a distributional extension of the
``parallel trends'' assumption, which is commonly used under the mean DID
framework to ensure that the \textit{path} of potential outcomes without
treatment are the same between the treated and untreated groups on average.
We extend this mean DID assumption to allow for the distribution of the path
to be the same for the treated and untreated groups. This Distributional DID
assumption is not strong enough to point-identify the counterfactual
distribution of outcomes for the treated group as well as the CQTT. This is
because identifying the counterfactual distribution for the treated group
hinges on (i) knowing the distribution of the change in untreated potential
outcomes and (ii) knowing the dependence, or copula, between the change and
initial level of untreated potential outcomes. The Distributional DID
assumption handles the first identification challenge but not the second.
The key innovation of the current paper is to introduce the Copula
Invariance assumption, which replaces the unknown copula with the
observed copula between the change and initial level of outcomes for the
untreated group and handles the second identification challenge.
Importantly, the Copula Invariance assumption allows for the marginal distribution of
untreated outcomes in the period before treatment to differ for the treated
and control groups. Moreover, this assumption only requires that the
researcher has access to two periods of data.

Given the point-identification result for the counterfactual distribution,
we propose a two-step estimation procedure based on conditional empirical
distributions. In the first step, we estimate conditional empirical
distributions of observed outcomes for the treated and untreated groups
separately. In the second step, the first-step estimates are used in the
estimation of the distribution of conditional counterfactual outcomes. The
CQTT is estimated by inverting the estimated counterfactual distribution of
potential outcome for the treated group. The proposed estimator is shown to
converge in distribution to a Gaussian process at the parametric rate
through empirical process techniques, while the limiting process is not
nuisance parameter free because of estimation error at the first step and
our identification strategy. To obtain an accurate approximation for the
limiting process in a finite sample, we consider the exchangeable bootstrap
proposed by \cite{praestgaard-wellner-1993} and show its first order
validity.

This paper contributes to a growing literature on estimating quantile
treatment effects (QTE) with observational data. The conditional QTE can be
analyzed under the quantile regression framework proposed by \cite%
{koenker-bassett-1978}. See \cite{koenker-2005} for a comprehensive review. \cite%
{abadie-angrist-imbens-2002} and \cite{chernozhukov-hansen-2005, CH2006JoE,
ChernozhukovHansen2013ARE} consider the conditional QTE allowing for
endogeneous regressors in quantile regression when instrumental variables
are available. For unconditional QTE, \cite{firpo-2007} proposes a
propensity score weighting estimator under a selection on observables
assumption and \cite{abadie-2003} and \cite{frolich-melly-2013} consider the
case where a researcher has access to an instrument that satisfies the
exclusion restriction only after conditioning on some covariates. Also, a
counterfactual distribution of an outcome under a counterfactual covariates
distribution has been studied by \cite{firpo-fortin-lemieux-2009}, \cite%
{Rothe2012Emtca} and \cite{chernozhukov-val-melly-2013} among others.

\cite{athey-imbens-2006} suggest the Changes in
Changes (CIC) model as an alternative to the DID model and \cite%
{melly-santangelo-2015} extend it to the case where the identifying
assumptions hold conditional on covariates. \cite{ChernozhukovEtAl2013Emtca}
consider identification of the conditional ATE and QTE for nonseparable
panel data models under a time-homogeneity condition. \cite%
{Dhaultfoeuille2015} present identification of nonseparable models using
repeated cross-sections. Other recent contributions to panel quantile
regression include 
\cite{AD2008JBES}, 
\cite{canay-2011}, 
\cite{KatoGalvaoMontesrojas2012JoE}, 
\cite{rosen-2012}, 
\cite{GLL2013JASA}, 
\cite{GalvaoKato2014JBES}, 
\cite{CFHHN2015JoE}, 
\cite{graham2015quantile}, 
\cite{LiOka2015JoE},
\cite{arellano2016nonlinear},
\cite{CLP2016Emtca} and
\cite{ShakeebEtAl2016JoE} 
among others.
 
Our work is closely related to \cite{fan-yu-2012} and \cite{callaway-li-2015}
in that both works consider the identification of the QTT under a
Distributional DID assumption. \cite{fan-yu-2012} consider the partial
identification of the QTT by using a Frechet-Hoeffding bound for a copula,
while \cite{callaway-li-2015} establish the point-identification of the QTT
using at least three periods of panel data by replacing an unknown,
unconditional copula in the last two periods with an observed copula in the
first two periods. In this paper, we focus on estimating conditional QTEs,
while it would be relatively straightforward to extend our analysis to
accommodate the unconditional QTE. Conditioning on covariates makes the
Distributional DID assumption more likely to hold in empirical applications.
For example, it seems likely that the path of earnings in the absence of
treatment depends on individual characteristics such as age and experience.
If these covariates are distributed differently across the treated and
untreated groups, then an unconditional DID approach will not be valid
though a conditional DID approach will %
\citep{heckman-ichimura-smith-todd-1998, abadie-2005}. Similarly, our Copula
Invariance assumption is more plausible when it holds conditional on
covariates. We focus on identifying and estimating the CQTT in the case where all covariates are discrete which allows
for nonparametric estimation that does not suffer from the curse of
dimensionality and is a relevant case for much applied research. The
identification and estimation with discrete regressors have been considered
in \cite{ChernozhukovEtAl2013Emtca} for the conditional ATE and QTE and \cite%
{graham2015quantile} for quantile panel data model with random coefficients.

We use our method to investigate the effect of increases in the minimum wage
on the distribution of earnings. Using individual level panel data from the
Current Population Survey, we consider a group of five treated states that
increased their minimum wage in 2007 relative to a control group of states
that did not increase the minimum wage. Forming groups based on race,
gender, and education, we are only able to reject the null of no effect at
any quantile for three out of eight groups. For these groups, the effect of
the minimum wage is concentrated in the lower part of the distribution and
appears to result in lower earnings.

The paper is organized as follows. In Section 2, we provide identification
results for the CQTT under the Distributional DID assumption and the Copula
Invariance assumption when both hold after conditioning on some covariates.
In Section 3, we discuss an estimation procedure and provide asymptotic
results. We also introduce a resampling procedure and give a theoretical
result on its consistency. Section 4 investigates the finite sample
performance of our estimator using Monte Carlo simulations. In Section 5, we
present an empirical application on the effect of increasing the minimum
wage on the distribution of earnings. Section 6 concludes.
All proofs are given in the Appendix.

\section{Identification}

This section considers the main identification results in the paper. We
consider identification of the QTT when the identifying assumptions hold
after conditioning on some observed covariates $X$. We also consider
additional requirements for identification when only repeated cross sections
are available. Through the paper, we use $\supp(V)$ and $\supp(V|W)$ to
denote the support of $V$ and the support of $V$ conditional on $W$,
respectively, for some random variables $V$ and $W$.

\subsection{Identification}

We consider a DID framework, in which all individuals in the sample receive
no treatment in period $t-1$ while a fraction of individuals receive the
treatment in period $t$. Let $D_{it}$ be a treatment indicator that takes
the value one if individual $i$ is treated in period $t$ and zero
otherwise. For each individual $i$, there is a pair of potential outcomes $%
(Y_{is}(0), Y_{is}(1) )$ in period $s \in\{t-1, t\}$, where $Y_{is}(0)$ and $%
Y_{is}(1)$ denote potential outcomes in the untreated and treated state in
period $s$, respectively. Every individual experiences either treated or
untreated status but not both, and thus the pair of potential outcomes is
not observable.

We suppose that researchers can access panel data, which consist of
outcomes, covariates and treatment statuses for each individual over some
periods including both before and after the implementation of the program of
interest. We consider the case with two periods of panel data and so that the
data consists of observations $\{(Y_{i,t-1}, Y_{it}, X_{i},
D_{it})\}_{i=1}^{n}$ with $n$ denoting the sample size, where the observed
outcomes are given by 
\begin{eqnarray}  \label{eq:outcome-def}
Y_{i, t-1}:= Y_{i, t-1}(0) \ \ \ \mathrm{and} \ \ \ Y_{it}:= (1 - D_{it})
Y_{it}(0) + D_{it} Y_{it}(1).
\end{eqnarray}
Also, $X_{i}$ denotes a vector of covariates, which may include all time-varying
variables as well as time-invariant variables. Throughout the paper, we
assume independent and identically distributed (i.i.d.) observations within
treatment and control group as stated below.

\vspace{0.5cm} \noindent \textbf{Assumption A1} (Random sampling). The
two-periods panel data consists of observations $\{(Y_{i,t-1} , Y_{it},
X_{i}, D_{it})\}_{i=1}^{n}$ from the structure in (\ref{eq:outcome-def}).
The potential outcomes $\big (Y_{i,t-1}(0), Y_{i,t-1}(1) \big)$ and $\big (%
Y_{it}(0), Y_{it}(1) \big)$ are cross-sectionally i.i.d.~conditional on
treatment status $D_{it}$. \vspace{0.5cm} 

This assumption allows for the possibility that the marginal or joint
distributions of potential outcomes can be different between treatment and
control groups. Under this conditional random sampling assumption, we use $%
F_{Y_{s}(0)|X, D_{t}}$ and $F_{Y_{s}(1)|X,D_{t}}$ to denote conditional
distributions of the potential outcomes in period $s$ given covariates and
treatment status and let $F_{Y_{s}|X, D_{t}}$ be the conditional distribution of the 
observed outcome $Y_{is}$ in period $s$ given covariates and
treatment status.

Our primary goal is to identify distributional features of treatment effects
through conditional distributions of observed outcomes given covariates and $%
D_{it}$. The issue of identification of treatment effects arises because the
pair of potential outcomes is unobservable for each individual and thus
marginal distributions of potential outcomes are not necessarily identified
from data. For instance, the conditional distribution $F_{Y_{t}(1)|X,
D_{t}=1}$ of a potential outcome $Y_{it}(1)$ given $X_{i}$ and $D_{it}=1$
can be identified by the observed distribution $F_{Y_{t}|X, D_{t}=1}$ of
outcomes for the treated group. But the counterfactual distribution
$F_{Y_{t}(0)|X, D_{t}=1}$  of untreated
potential outcomes for the treated group cannot be identified generally
from the sample. Thus, for identifying distributional features of treatment
effects, we need to make additional restrictions.

As a measure of treatment effects, this paper considers the identification
and estimation of the CQTT given $X_{i} = x$, which measures the quantile treatment effect
within a subpopulation of individuals with treatment status $D_{it}=1$ and
common history $X_{i} =x$. Let $\mathcal{X}$ be the common support of $X_{i}$
for untreated and treated groups. The CQTT given $x \in \mathcal{X%
}$ at $\tau \in \mathcal{T} \subset (0,1)$ is defined as 
\begin{eqnarray*}
\Delta_{x}^{QTT}(\tau) := F_{Y_{t}(1)|X=x, D_{t}=1}^{-1}(\tau) -
F_{Y_{t}(0)|X=x, D_{t}=1}^{-1}(\tau),
\end{eqnarray*}
where $F_{Y_{t}(j)|X, D_{t}}^{-1}$ is the quantile
function of $Y_{it}(j)$
conditional on $X_{i}$ and $D_{it}$, given by 
\begin{eqnarray*}
F_{Y_{t}(j)|X, D_{t}}^{-1}(\tau) := \inf \big \{ y \in \mathbb{R}:
F_{Y_{t}(j)|X , D_{t}}(y) \ge \tau \big \}, \ \ \ j = 0, 1.
\end{eqnarray*}
As discussed in the preceding paragraph, we can identify the distribution $%
F_{Y_{t}(1)|X=x, D_{t}=1}$ from observables. For identifying the CQTT, it
remains to establish the identification of the counterfactual distribution $%
F_{Y_{t}(0)|X=x, D_{t}=1}$. To this end, we need to make three additional
restrictions.

The first condition restricts a time-difference of potential outcomes
without treatment, denoted by $\Delta Y_{it}(0):=Y_{it}(0)-Y_{i,t-1}(0)$, is
independent of the treatment status $D_{it}$ conditional on $X_{i}$.

\vspace{0.5cm} \noindent \textbf{Assumption A2.} (Distributional Difference
in Differences) 
\begin{eqnarray*}
\Pr \big \{ \Delta Y_{it}(0) \le \Delta y |X_{i}, D_{it} = 1 \big \} = \Pr %
\big \{ \Delta Y_{it}(0) \le \Delta y |X_{i}, D_{it} = 0 \big \},
\end{eqnarray*}
for all $\Delta y \in \supp \big( \Delta Y_{it}(0)|X_{i}\big)$. \vspace{0.5cm%
} 

This assumption ensures that potential outcomes without treatment are
comparable between treatment and control groups after taking a
time-difference, conditional on covariates. An analogous condition employed under the mean DID
framework is the ``parallel trends'' assumption: 
\begin{eqnarray*}
E[\Delta Y_{it}(0)| X_{i}, D_{it} = 1] = E [\Delta Y_{it}(0) | X_{i}, D_{it}
= 0 ],
\end{eqnarray*}
which is necessary for identifying the ATT \citep[see][]{heckman-ichimura-smith-todd-1998, abadie-2005}. The
distributional restriction in Assumption A2 replaces this standard mean
restriction. If multiple pre-treatment periods in sample are available, then
this assumption can be tested under a strict stationary assumption.

For the treated group, we can identify (i) the distribution $%
F_{Y_{t-1}(0)|X, D_{t}=1}$ of untreated potential outcomes in period $t-1$
from observed outcomes and (ii) the distribution $F_{\Delta Y_{t}(0)|X,
D_{t}=1}$ of the change in untreated potential outcomes through Assumption
A2 (the distributional DID). When these two distributions are identified,
the average untreated potential outcome (and hence, the ATT) is identified.
Without imposing an additional assumption, however, the CQTT is not
point-identified, while the CQTT can be partially identified along the line of 
\cite{fan-yu-2012}. This is because many possible distributions of untreated
potential outcomes in period $t$ are observationally equivalent. For
example, the distribution $F_{Y_{t}(0)|X, D_{t} =1}$of untreated potential
outcomes in period $t$ will be highly unequal if the change in untreated
potential outcomes and the initial untreated outcome are strongly positively
dependent. On the other hand, the distribution of untreated potential
outcomes in period $t$ will be less unequal if the change and initial level
of untreated potential outcomes are independent or negatively dependent.

The next condition imposes a restriction on the joint distribution $%
F_{\Delta Y_{t}(0), Y_{t-1}(0) |X, D_{t}}$ of $\Delta Y_{it}(0)$ and $%
Y_{i,t-1}(0)$ conditional on $X_{i}$ and $D_{it}$ through the copula $%
C_{\Delta Y_{t}(0), Y_{t-1}(0) |X, D_{t}}$ of $\Delta Y_{it}(0)$ and $%
Y_{i,t-1}(0)$ conditional on $X_{i}$ and $D_{it}$. By Sklar's theorem, we
have 
\begin{eqnarray*}
F_{\Delta Y_{t}(0), Y_{t-1}(0) |X, D_{t}} (\Delta y, y) = C_{\Delta
Y_{t}(0), Y_{t-1}(0) |X, D_{t}} \big ( F_{ \Delta Y_{t}(0) |X, D_{t}}
(\Delta y), F_{Y_{t-1}(0) |X, D_{t}} (y) \big ),
\end{eqnarray*}
for $(\Delta y, y) \in \supp(\Delta Y_{it}(0), Y_{i,t-1}(0) |X_{i}, D_{it})$%
. The following condition requires an invariance of the conditional copula
with respect to the treatment status $D_{it}$ conditional on $X_{i}$.

\vspace{0.5cm} \noindent \textbf{Assumption A3} (Copula Invariance). 
For each $x \in \mathcal{X}$ and for all $(u, v) \in [0,1]^{2}$,
\begin{eqnarray*}
C_{\Delta Y_{t}(0), Y_{t-1}(0) |X=x, D_{t}=1}(u, v) = C_{\Delta Y_{t}(0),
Y_{t-1}(0) |X=x, D_{t}=0}(u, v).
\end{eqnarray*}
\vspace{0.5cm} 

Given a realized value of some random variable, the marginal distribution
evaluated at this value can be interpreted as a ranking normalized to the
unit interval. The conditional copula function captures some rank dependency
between two variables $\Delta Y_{it}(0)$ and $Y_{i,t-1}(0)$ conditional on $%
X_{i}$ and $D_{it}$ and Assumption A3 requires that the dependency of
ranking of these random variables $\Delta Y_{it}(0)$ and $Y_{i,t-1}(0)$ are
the same for the treated and control groups. As in Assumption A1, however,
this assumption does not rule out the possibility that the joint
distribution of $\Delta Y_{it}(0)$ and $Y_{i,t-1}(0)$ conditional on $X_{i}$
and $D_{it}$ varies between the treatment and control group. Likewise,
Assumption A3 does not imply Assumption A2 because A3 restricts only the
copula, not the marginal distribution of the change in untreated potential
outcomes over time.

The Copula Invariance assumption recovers the missing dependence required to
uniquely identify the counterfactual distribution of untreated potential
outcomes for the treated group. It does so by replacing the unknown copula
for the treated group with the known copula from the untreated group.
Intuitively, if, for example, we observe that observations in the control
group at the top of the distribution of initial outcomes tend to experience
the largest increases in outcomes over time, the Copula Invariance
assumption implies that this would also occur for the treated group. The
Distributional DID assumption implies that the
distribution of the change in outcomes is the same for the two groups. But
the initial distribution of outcomes can be different for the two groups.

The Distributional DID assumption and the Copula
Invariance assumption are not directly testable. However, in the spirit of
placebo testing in DID models, they both can be tested
using periods before the treated group becomes treated. One idea would be to
directly test each assumption in these earlier periods. Another, simpler
idea would be to implement our procedure in the earlier periods and test
that $\Delta^{QTT}_x(\tau) = 0$ for all $\tau \in \mathcal{T}$.

As an additional identifying restriction, we assume continuity conditions on
distributions of potential outcomes and its time-difference as below.

\vspace{0.5cm} \noindent \textbf{Assumption A4} (Continuous distributions).
Each random variable of $\Delta Y_{it}(0)$ and $Y_{i,t-1}(0)$ has a
continuous distribution conditional on $X_{i}$ and $D_{it}$
and a random variable $Y_{it}(1)$ also has a continuous distribution
conditional on $X_{i}$ and $D_{it} = 1$. Each distribution has a compact
support with densities uniformly bounded away from 0 and $\infty$ over the
support. \vspace{0.5cm} 

The continuity of marginal distributions conditional on treatment status
guarantees that conditional copulas in Assumption A3 are unique and
facilitates the identification analysis. More precisely, we obtain
identification by employing the Rosenblatt transform, which is the
distribution transform studied by \cite{Rosenblatt1952AMS}. Also, Assumption
A4 imposes a compact support assumption as in \cite{athey-imbens-2006} in
order to avoid technical difficulties in the rest of analysis, while this
condition is not used for our identification analysis and can be replaced by
other conditions for the rest of the results.

Given the random sample in Assumption A1, the additional conditions in
Assumption A2-A4 deliver the point-identification of the counterfactual
distribution as in the following theorem.

\vspace{0.5cm}

\begin{theorem}
\label{theorem:identification-1}  Suppose that Assumption A1-A4 hold. Then,  
\begin{eqnarray*}
  F_{Y_{t}(0)|X=x, D_{t}=1} (y) = \Pr \big \{ \Delta Y_{it} + F_{ Y_{t-1}|X=x,
D_{t}=1}^{-1} \circ F_{Y_{t-1}|X=x, D_{t}=0} ( Y_{i,t-1} ) \le y | X_{i}{=}%
x, D_{it}{=}0 \big \},
\end{eqnarray*}
for all $x \in \mathcal{X}$ and $y \in \supp(Y_{it}(0)|X_{i}=x, D_{it}{=}1)$.
\end{theorem}

\vspace{0.5cm} 

The above theorem shows that the counterfactual distribution of interest can
be identified from observed outcomes of untreated individuals. This implies
that that treated and untreated groups must be similar in the distributional
sense of not only marginal distribution but also some dependency over
periods, and thus Assumption A2 and A3 play a crucial role for the
identification as shown in its proof. An immediate consequence is the
identification of the conditional QTT since the other distribution $%
F_{Y_{t}(1)|X=x, D_{t}=1}$ is identified by the distribution $F_{Y_{t}|X=x,
D_{t}=1}$ of observed outcomes.

Another implication of Theorem 1 is that unconditional QTTs are identified
using our approach. One can simply average over the covariates in the
conditional counterfactual distribution in Theorem 1 to obtain an
unconditional counterfactual distribution and then invert it to obtain the
unconditional quantiles. Thus, our method can be comparable to the results
that obtain unconditional QTTs; however, for the rest of the paper we focus
only on the CQTT.

As an extension of the result above, we establish the identification of the
counterfactual distribution when the available data consists of two-periods of repeated
cross sections, rather than panel data. 
In the following corollary, we only consider the case with time invariant covariates, 
following the DID literature with repeated cross sections \citep{abadie-2005, melly-santangelo-2015}.
Even when the data generating
process satisfies Assumption A1-A4, the change in outcome over time, $Y_{it}
- Y_{i,t-1}$, is unobserved because each individual in sample is observed
only in one period. 
To deal with the identification issue due to the data,
we consider a restriction of the conditional rank invariance, which enables us to
recover individual outcome in period $t$ by using the rank of outcome in
period $t-1$ as formalized by the following corollary.

\vspace{0.5cm}

\begin{corollary}
\label{corollary:identification-EX}  Consider the repeated cross sections  $%
\{(Y_{is}, X_{i}, D_{is})\}_{i=1}^{m^{(s)}}$  in period $s \in \{ t-1, t\}$ with $m^{(s)}$
being the sample size.  Suppose that the data generating
process for  the repeated cross sections satisfy Assumption A1-A4 hold.  
If the conditional copula of  $(Y_{i,t-1}(0),Y_{i, t}(0) )$  given 
$X_{i}=x$ and $D_{it}=1$
satisfies the rank invariance: for every $(u,v) \in [0,1]^2$, 
\begin{eqnarray*}
  C_{Y_{t-1}(0), Y_{t}(0)|X=x,D_{t}=1} (u, v) = \min\{u, v\},
\end{eqnarray*}
then,  for $y \in \supp(Y_{it}(0)|X_{i}=x, D_{it}=1)$,
\begin{eqnarray*}
  F_{Y_{t}(0)| X=x, D_{t}=1} (y) 
  = 
  \Pr \big \{ 
  \widetilde{\Delta Y}_{it}
  +
  F_{ Y_{t-1}|X=x,D_{t}=1}^{-1}  
  \circ F_{Y_{t-1}|X=x, D_{t}=0} ( Y_{i,t-1} ) 
  \le y | X_{i}=x, D_{it} = 0 \big \},
\end{eqnarray*}
where
$\widetilde{\Delta Y}_{it}:=
F_{ Y_{t}|X, D_{t}=0}^{-1} 
\circ 
F_{Y_{t-1}|X, D_{t}=0} ( Y_{i,t-1} ) - Y_{i, t-1}
$.
\end{corollary}

\vspace{0.5cm}


The additional assumption of conditional rank invariance says that, for observations with the same observed covariates, individuals maintain their rank in the distribution of outcomes over time.  This assumption is weaker than unconditional rank invariance as some individuals can change their rank in the distribution of earnings over time.  It does not imply nor is implied by the Copula Invariance assumption, nor does it imply conditional rank invariance between $\Delta Y_{it}(0)$ and $Y_{it-1}(0)$.


\section{Estimation and Inference}

As in the previous section, the counterfactual distribution is identified by
distributions of observed outcomes conditional on covariates and treatment
status. In this section, we first explain an estimation procedure based on
conditional empirical distributions and then provide asymptotic results for
the proposed estimator using the functional delta method. 
We develop uniform inference results using techniques from the literature on
empirical processes \citep[see, for example,][]{VW1996}.
In this section, we consider the case where all covariates are discrete, which allows
for nonparametric estimation that does not suffer from the curse of
dimensionality. The estimation with discrete regressors is also considered
in \cite{ChernozhukovEtAl2013Emtca} and \cite{graham2015quantile}.

\subsection{Estimation}

We estimate the conditional distribution $F_{Y_{s}|X=x, D_{t}=d}$ of
observed outcome $Y_{is}$ given covariates $X_{i}=x$ and treatment status $%
D_{it}=d$ by using the corresponding empirical distribution. For 
$d \in \{0, 1\}$, let $\delta_{i, x}^{(d)}:= 1\{X_{i} = x,
D_{it}=d\}$ and $n_{x}^{(d)} = \sum_{i=1}^{n}\delta_{i, x}^{(d)}$. Then,
the estimator of $F_{Y_{s}|X = x, D_{t}=d}$ is given by, for $s\in \{t-1,t\}$ and $%
d \in \{0, 1\}$, 
\begin{eqnarray}  \label{eq:CDF-0}
\hat{F}_{Y_{s}|X = x, D_{t}=d}(y) := \frac{ 1 }{ n_{x}^{(d)} }
\sum_{i=1}^{n} 1 \big \{ Y_{is} \le y \big \} \delta_{i, x}^{(d)}.
\end{eqnarray}
We denote an estimator for $F_{Y_{t}(1)|X=x, D_{t}=1}(y)$ by $\hat{F}%
_{Y_{t}(1)|X=x, D_{t}=1}(y)$, which is given by the empirical distribution $%
\hat{F}_{Y_{t}|X=x, D_{t}=1}(y)$ because $Y_{it} = Y_{it}(1)$ if $D_{it}=1$.
For estimation of the counterfactual distribution provided in Theorem \ref%
{theorem:identification-1}, we obtain estimated quantiles $\hat{F}_{
Y_{t-1}|X=x, D_{t}=1}^{-1}$ from the empirical distribution $\hat{F}_{
Y_{t-1}|X=x, D_{t}=1}$ and then set 
\begin{eqnarray}  \label{eq:est-main}
\hat{F}_{Y_{t}(0)|X = x, D_{t}=1}(y) := \frac{ 1 }{ n_{x}^{(0)} }
\sum_{i=1}^{n} 1 \big \{ \Delta Y_{it} + \hat{F}_{ Y_{t-1}|X=x,
D_{t}=1}^{-1} \circ \hat{F}_{Y_{ t-1}|X = x, D_{t}=0} ( Y_{i, t-1} ) \le y %
\big \} \delta_{i, x}^{(0)},
\end{eqnarray}
for $y \in \mathbb{R}$. We use estimated distribution functions $\hat{F}_{
Y_{t}(1)|X=x, D_{t}=1}$ and $\hat{F}_{ Y_{t}(0)|X=x, D_{t}=1}$ to obtain
quantiles for each distribution. Then, the CQTT estimator is given by 
\begin{eqnarray*}
\hat{\Delta}_{x}^{QTT}(\tau) := \hat{F}_{Y_{t}(1)|X=x, D_{t}=1}^{-1}(\tau)
- \hat{F}_{Y_{t}(0)|X = x, D_{t}=1}^{-1}(\tau),
\end{eqnarray*}
for $(\tau, x) \in \mathcal{T} \times \mathcal{X}$.

\subsection{Asymptotic Results}

We provide a functional central limit theorem for the CQTT estimator over $%
\mathcal{T}$, where $\mathcal{T}$ is assumed to be a compact subset strictly
within the unit interval. We begin with a preliminary result on weak
convergence of empirical distributions, which facilitates the use of the
functional delta method with Hadamard differentiable maps. In what follows,
we denote by $\mathcal{Y}_{s|x, d}:=\supp(Y_{is}|X_{i}=x, D_{it}=d)$ and $%
\mathcal{Y}_{s|x, 1}(j):=\supp(Y_{is}(j)|X_{i}=x, D_{it}=1)$ for $s\in\{t-1, t\}$, 
$d \in\{ 0,1\}$ and $j \in \{0,1\}$.

For each $(s, d, x) \in \{t-1, t\} {\times} \{0, 1\} {\times} \mathcal{X}$, define empirical
processes as 
\begin{eqnarray*}
\hat{G}_{s, x}^{(d)}(y) := \sqrt{n} \big ( \hat{F}_{Y_{s}|X=x, D_{t}=d}(y) -
F_{Y_{s}|X=x, D_{t}=d}(y) \big ), \ \ \ y \in \mathcal{Y}_{s|x, d}.
\end{eqnarray*}
Let $ \tilde{Y}_{it}  :=  \Delta Y_{it}  +  F_{ Y_{t-1}|X= x, D_{t}=1}^{-1} 
\circ  F_{Y_{t-1}|X=x, D_{t}=0}  ( Y_{i,t-1} ) $ and we use $\tilde{F}%
_{Y_{t}(0)|X=x, D_{t}=1}(y)$ to denote the infeasible estimator for the
counterfactual distribution based on observations $\{\tilde{Y}_{it}\}$ with $%
X_{i}=x$ and $D_{it}=1$ as in (\ref{eq:CDF-0}). Define its empirical
process as 
\begin{eqnarray}  \label{eq:ep-tilde}
\tilde{G}_{t, x}^{(0)}(y) := \sqrt{n} \big ( \tilde{F}_{Y_{t}(0)|X=x,
D_{t}=1}(y) - F_{Y_{t}(0)|X=x, D_{t}=1}(y) \big ), \ \ \ \ \ y \in \mathcal{Y%
}_{t|x, 1}(0).
\end{eqnarray}
We make an additional assumption.

\vspace{0.5cm} \noindent \textbf{Assumption A6}. (a) A pair of random
variables $(\Delta Y_{it}, Y_{i,t-1})$ is continuously distributed
conditional on $X_{i}$ and $D_{it} = 0$ over a compact support with a
distribution $F_{\Delta Y_{t}, Y_{t-1}|X, D_{t} = 0}$ and a density $%
f_{\Delta Y_{t}, Y_{t-1}|X, D_{t} = 0}$. A random variable $\Delta Y_{it}$
is continuously distributed conditional on $Y_{i, t-1}$, $X_i$, and $D_{it}
= 0$ with a uniformly continuous density $f_{\Delta Y_{t}|Y_{t-1}, X,
D_{t}=0}$ over a compact support. (b) The sample sizes $n_{x}^{(0)}$ and $%
n_{x}^{(1)}$ go to $\infty$ as $n \to \infty$, while $r_{x}^{(j)}:= \lim_{n
\to \infty} (n /n_{x}^{(j)})^{1/2} \in (0, \infty)$ for $j = 0, 1$. \vspace{%
0.5cm} 

The following lemma provides a functional central limit theorem for the
empirical processes above. We define $\mathbb{S}_{x}  :=  \ell^{\infty}  %
\big( \mathcal{Y}_{t|x,1}(0) \big)
{\times}  \ell^{\infty}  ( \mathcal{Y}_{t-1|x, 0} )  {\times}  \ell^{\infty}
( \mathcal{Y}_{t|x, 1} )  {\times}  \ell^{\infty}  ( \mathcal{Y}_{t-1|x, 1}
) $ for a fixed $x \in \mathcal{X}$, where $\ell^{\infty}(S)$ denotes the
space of all uniformly bounded functions on some set $S$, equipped with
supremum norm $\|\cdot\|_{\infty}$.

\vspace{0.5cm}

\begin{lemma}
\label{lemma:asym-basic}  Suppose that Assumption A1-A6 hold. Then,  for
each $x \in \mathcal{X}$,  
\begin{eqnarray*}
\big ( \tilde{G}_{t, x}^{(0)}, \hat{G}_{t-1, x}^{(0)}, \hat{G}_{t, x}^{(1)}, 
\hat{G}_{t-1, x}^{(1)} \big ) \rightsquigarrow \big ( \mathbb{V}_{x}^{(0)}, 
\mathbb{W}_{x}^{(0)}, \mathbb{V}_{x}^{(1)}, \mathbb{W}_{x}^{(1)} \big ),
\end{eqnarray*}
in the space  $\mathbb{S}_{x}$.  Here,  $ \big (
\mathbb{V}_{x}^{(0)},  \mathbb{W}_{x}^{(0)},  \mathbb{V}_{x}^{(1)},  \mathbb{%
W}_{x}^{(1)}  \big )
$  is a tight Gaussian process with mean zero and covariance kernel  $%
\mathrm{diag}\{\Sigma_{x}^{(0)}(\cdot, \cdot), \Sigma_{x}^{(1)}(\cdot,
\cdot)\}$  defined on $\mathbb{S}_{x}$,  where  $\Sigma_{x}^{(j)}(\cdot,
\cdot)$  is the $2 \times 2$  positive definite,  covariance kernel of  $(%
\mathbb{V}^{(j)}, \mathbb{W}^{(j)})$  for $j=0,1$,  given by,  for $(y_{1},
y_{2}, y_{3}, y_{4}) \in \mathbb{S}_{x}$,  
\begin{eqnarray*}
\Sigma_{x}^{(0)}(y_{1}, y_{2}) := \var_{x}^{(0)} \big ( I_{it}^{(0)} (y_{1},
y_{2}) \big ) \ \ \ \ \mathrm{and} \ \ \ \ \Sigma_{x}^{(1)}(y_{3}, y_{4}) := %
\var_{x}^{(1)} \big ( I_{it}^{(1)}(y_{3}, y_{4}) \big ),
\end{eqnarray*}
with  $I_{it}^{(0)}(y_{1}, y_{2}):=  (1\{ \tilde{Y}_{it} \le y_{1}\},1\{
Y_{i,t-1}\le y_{2}\})^{\prime } $,  $I_{it}^{(1)}(y_{3}, y_{4}):=  (  1\{
Y_{it} \le y_{3}\},  1\{ Y_{i,t-1}\le y_{4}\}  )^{\prime } $
and  $\var%
_{x}^{(j)}$ being the covariance function conditional on  $X_{i} =x$  and  $%
D_{it} = j$.
\end{lemma}

\vspace{0.5cm} 

Using the result in this lemma, we first obtain the joint limiting process
for the estimator of the potential outcome distributions
$(\hat{F}_{Y_{t}(0)| X=x, D_{t}=1}, \hat{F}_{Y_{t}(1)|X=x, D_{t}=1})$. It is
straightforward from the above result to obtain the limit process for the
distribution $\hat{F}_{Y_{t}(1)| X=x, D_{t}=1}$, which is identified
directly from data; whereas, the one for the counterfactual distribution $%
\hat{F}_{Y_{t}(0)| X=x, D_{t}=1}$ needs several steps. Since the estimator
for the counterfactual distribution in (\ref{eq:est-main}) can be considered
as a process indexed by functions depending on estimated distributions, we
use recent results for empirical processes in \cite{VW2007IMS} with some
modifications in order to obtain the limiting process as formalized by the
following proposition.

\vspace{0.5cm}

\begin{proposition}
\label{proposition:F-convergence}  Define  $\hat{Z}_{x}^{(j)}(y)  :=  \sqrt{n%
}  \big (
\hat{F}_{Y_{t}(j)|X=x, D_{t}=1}(y)  -  F_{Y_{t}(j)|X=x, D_{t}=1}(y)  \big )
$  for each  $x \in \mathcal{X}$,  $j = 0,1$  and  $y \in \mathcal{Y}_{t|x,
1}(j)$.  Suppose that Assumption A1-A6 hold.  Then,  
\begin{eqnarray*}
\big ( \hat{Z}_{x}^{(0)}, \hat{Z}_{x}^{(1)} \big ) \rightsquigarrow \big ( 
\mathbb{Z}_{x}^{(0)}, \mathbb{Z}_{x}^{(1)} \big ),
\end{eqnarray*}
in the metric space  $\ell^{\infty}\big ( \mathcal{Y}_{t|x, 1}(0) \big)
{\times}  \ell^{\infty}\big ( \mathcal{Y}_{t|x, 1}(1) \big)
$.  Here,  $(\mathbb{Z}_{x}^{(0)}, \mathbb{Z}_{x}^{(1)})$  is a tight
zero-mean Gaussian process with a.s.~uniformly continuous  paths on $%
\mathcal{Y}_{t|x, 1}(0) {\times} \mathcal{Y}_{t|x, 1}(1)$, given by  
\begin{eqnarray*}
\mathbb{Z}_{x}^{(0)} := r_{x}^{(0)}\mathbb{V}_{x}^{(0)} + \kappa_{x}( 
\mathbb{W}_{x}^{(0)}, \mathbb{W}_{x}^{(1)} ) \ \ \ \ \ \mathrm{and} \ \ \ \ \
\mathbb{Z}_{x}^{(1)} = r_{x}^{(1)} \mathbb{V}_{x}^{(1)},
\end{eqnarray*}
where  the map  $\kappa_{x}:  \ell^{\infty}(\mathcal{Y}_{t-1|x,0})  {\times}
\ell^{\infty}(\mathcal{Y}_{t-1|x,1})  \mapsto  \ell^{\infty}(\mathcal{Y}%
_{t|x,1}(0))  $ is given by  
\begin{eqnarray*}
\kappa_{x}(W_{0}, W_{1}) := \int \big \{ r_{x}^{(0)} W_0(v) - r_{x}^{(1)}
W_{1} \circ F_{Y_{t-1}|X=x, D_{t}=1}^{-1} \circ F_{Y_{t-1}|X=x, D_{t}=0}(v) %
\big \} \omega_{x}(y,v) d v,
\end{eqnarray*}
for  $(W_{0}, W_{1}) \in  \ell^{\infty}(\mathcal{Y}_{t-1|x,0})  {\times} 
\ell^{\infty}(\mathcal{Y}_{t-1|x, 1})  $  with  
\begin{eqnarray*}
\omega_{x}(y, v):= \frac{ f_{\Delta Y_{t}, Y_{t-1}|X=x, D_{t}=0} \big ( y -
F_{Y_{t-1}|X=x, D_{t}=1}^{-1} \circ F_{Y_{t-1}|X=x, D_{t}=0}(v), v \big ) }{
f_{Y_{t-1}|X=x, D_{t}=1} \circ F_{Y_{t-1}|X=x, D_{t}=1}^{-1} \circ
F_{Y_{t-1}|X=x, D_{t}=0}(v) },
\end{eqnarray*}
for  $(y, v) \in \mathcal{Y}_{t|x,1}(0) {\times} \mathcal{Y}_{t-1|x,0}$.
\end{proposition}

\vspace{0.5cm} 

This proposition shows that the limiting process $\mathbb{Z}_{x}^{(0)}$ for
the counterfactual distribution has an extra term depending on the map $%
\kappa_{x}$, which reflects our identification argument of the
counterfactual distribution of interest as well as the contribution of
estimation errors from empirical distributions. Thus the limiting
distribution is not nuisance parameter free, and a bootstrap procedure can
facilitate statistical inference in practice as shown in the next subsection.

Next we present the limiting process of the CQTT estimators over a range of
quantile $\mathcal{T}$. Proposition \ref{proposition:F-convergence} together
with the functional delta method delivers the following theorem.

\vspace{0.5cm}

\begin{theorem}
\label{theorem:Q-convergence}  Suppose that Assumption A1-A6 hold.  If  $%
F_{Y_{t}(0)|X, D_{t}=1}$  admits a positive continuous density  $%
f_{Y_{t}(0)|X, D_{t}=1}$  on an interval $[a, b]$ containing an $\epsilon$%
-enlargement  of the set  $\{ F_{Y_{t}(0)|X, D_{t}=1}^{-1}(\tau): \tau \in 
\mathcal{T} \}  \subset \mathcal{Y}_{t|X,1}(0)$  with $\mathcal{T} \subset
(0,1)$,  then, for each $x \in \mathcal{X}$,  
\begin{eqnarray*}
\sqrt{n} \big ( \hat{\Delta}_{x}^{QTT}(\tau) - \Delta_{x}^{QTT}(\tau) \big ) %
\rightsquigarrow \bar{\mathbb{Z}}_{x}^{(1)}(\tau) - \bar{\mathbb{Z}}%
_{x}^{(0)}(\tau),
\end{eqnarray*}
where  $\big (  
\bar{\mathbb{Z}}_{x}^{(0)}(\tau),  \bar{\mathbb{Z}}_{x}^{(1)}(\tau)  \big )
$  is a stochastic process in the metric space  $(\ell^{\infty}(\mathcal{T})
)^2$,  given by  
\begin{eqnarray*}
\bar{\mathbb{Z}}_{x}^{(j)}(\tau) := \frac{ \mathbb{Z}_{x}^{(j)} \big ( %
F_{Y_{t}(j)|X=x, D_{t}=1}^{-1}(\tau) \big ) }{ f_{Y_{t}(j)|X=x, D_{t}=1} %
\big ( F_{Y_{t}(j)|X=x, D_{t}=1}^{-1}(\tau) \big ) },
\end{eqnarray*}
for $j = 0, 1$.
\end{theorem}

\vspace{0.5cm} 

Using the result in Proposition \ref{proposition:F-convergence}
with a similar argument used in Theorem \ref{theorem:Q-convergence},
one could also consider other plug-in estimators of
Hadamard differentiable functionals, such as Lorenz curve and Gini
coefficient and obtain their limit processes. 
Also, we can consider testing that $\Delta_x^{QTT}(\tau) = 0$ for all $\tau \in 
\mathcal{T}$ using the Kolmogorov-Smirnov (KS) test statistic given by 
$
KS_x := \sqrt{n} \sup_{\tau \in \mathcal{T}} \big | \hat{\Delta}%
^{QTT}_x(\tau) \big |.
$
The next corollary states this result; it follows directly from Theorem 2
and the continuous mapping theorem.

\vspace{0.5cm}

\begin{corollary}\label{corollary:ks}
Suppose that the conditions of Theorem 2 hold.  Under the null hypothesis $%
H_0: \Delta^{QTT}_x(\tau) = 0$ for all $\tau \in \mathcal{T}$,  we have, for
each $x \in \mathcal{X}$,  
\begin{align*}
KS_x \xrightarrow{d} \sup_{\tau \in \mathcal{T}} \big | \bar{\mathbb{Z}}%
_{x}^{(1)}(\tau) - \bar{\mathbb{Z}}_{x}^{(0)}(\tau) \big |.
\end{align*}
\end{corollary}

\vspace{0.5cm} 

In addition to testing for zero CQTT, one can also use this result to form asymptotic simulataneous $(1-\alpha)$\% confidence bands for the entire CQTT process.  The confidence bands are given by $( \hat{\Delta}^{QTT}_x - c_{1-\alpha}n^{-1/2}, \hat{\Delta}^{QTT}_x + c_{1-\alpha}n^{-1/2} )$ where $c_{1-\alpha}$ are critical values from the KS test.  In practice, the critical values can be obtained using the bootstrap.


\subsection{Bootstrap}

The limiting processes presented in the preceding section depend on unknown
nuisance parameters, some of which require nonparametric estimation and may
complicate inference in finite samples. To deal with the issue of nonpivotal
limit processes, we consider a resampling method called the exchangeable
bootstrap (see \cite{praestgaard-wellner-1993} and \cite{VW1996}). This
resampling scheme consistently estimates limit laws of relevant empirical
distributions and thus with the functional delta method consistently
estimates the limit process of the CQTT estimator.

For the resampling scheme, we introduce a vector of random weights $%
(W_{1}^{(d)}, \dots, W_{n}^{(d)})$ for $d=0 ,1$. To establish the validity
of the bootstrap, we assume that the random weights satisfy the following
conditions.

\vspace{0.5cm} \noindent \textbf{Assumption B}. For each $d \in \{0, 1\}$,
let $(W_{1}^{(d)}, \dots, W_{n}^{(d)})$ be an $n$-dimensional vector of
exchangeable, nonnegative random variables. The vectors $(W_{1}^{(0)},
\dots, W_{n}^{(0)})$ and $(W_{1}^{(1)}, \dots, W_{n}^{(1)})$ are independent
of the original sample as well as each other. The vectors of random weights,
depending on the size of each group, satisfy the following conditions: 
\begin{eqnarray*}
\max_{1 \le i \le n} E \big |W_{i}^{(d)} \big |^{2+\epsilon} < \infty, \ \ 
\overline{W}_{n, x}^{(d)} := \frac{ 1 }{ n_{x}^{(d)} } \sum_{i =1}^{n}
W_{i}^{(d)} \delta_{i,x}^{(d)} \to^p 1, \ \ \frac{ 1 }{ n_{x}^{(d)} }
\sum_{i =1}^{n} ( W_{i}^{(d)} - \overline{W}_{n, x}^{(d)} )^2
\delta_{i,x}^{(d)} \to^p 1, \ \ 
\end{eqnarray*}
for each $d \in \{0, 1\}$. \vspace{0.5cm} 

As \cite{VW1996} explain, this resampling scheme encompasses a variety of
bootstrap methods, such as the empirical bootstrap, subsampling, wild
bootstrap and so on. This condition is employed in \cite%
{chernozhukov-val-melly-2013} for inference of counterfactual distributions.
For the empirical application in this paper, we consider the empirical
bootstrap, which corresponds to the case where the vector of random weights $%
(W_{1}^{(d)}, \dots, W_{n}^{(d)})$ follows the multinomial distribution with
probabilities $\delta_{i,x}^{(d)} \cdot (1/n_{x}^{(d)}, \dots,
1/n_{x}^{(d)})$. Given each realization of random weights, we apply the estimation
procedure explained in the previous section and estimate the parameters of
interest. As we show below, the repetition of the bootstrap leads to
asymptotically valid inference. The other types of resampling methods, such
as weighted bootstrap or subsampling, also can be considered under the same
framework and shown to be valid. In our empirical application, the sample
size is moderate and our estimation procedure uses empirical distribution
functions and thus the empirical bootstrap is straightforward and
convenient. In the other applications, however, one might prefer weighted
bootstrap if the estimation procedure is time-consuming or subsampling if
the sample size is extremely large.

Given the random weights, we define the weighted bootstrap empirical
distribution as 
\begin{eqnarray*}
\hat{F}_{Y_{s}|X = x, D_{t}=d}^{\ast}(y) := \frac{ 1 }{ n_{x}^{(d)} }
\sum_{i=1}^{n} W_{i}^{(d)} 1 \big \{ Y_{is} \le y \big \} \delta_{i,
x}^{(d)},
\end{eqnarray*}
for $(s, d) \in \{ t-1, t\} \times \{ 0, 1 \}$. As in the previous subsection, the
bootstrap distribution of the treated potential outcome $\hat{F}%
_{Y_{t}(1)|X=x, D_{t}=1}^{\ast}$ is given by $\hat{F}_{Y_{t}|X=x,
D_{t}=1}^{\ast}$, while the bootstrap version of the counterfactual
distribution is given by 
\begin{eqnarray*}
\hat{F}_{Y_{t}(0)|X = x, D_{t}=1}^{\ast}(y) := \frac{ 1 }{ n_{x}^{(0)} }
\sum_{i=1}^{n} W_{i}^{(0)} 1 \big \{ \Delta Y_{it} + \hat{F}_{ Y_{t-1}|X=x,
D_{t}=1}^{ \ast -1} \circ \hat{F}_{Y_{ t-1}|X = x, D_{t}=0}^{\ast} ( Y_{i,
t-1} ) \le y \big \} \delta_{i, x}^{(0)},
\end{eqnarray*}
for $y \in \mathbb{R}$, where $\hat{F}_{ Y_{t-1}|X=x, D_{t}=1}^{\ast -1}$ is
the bootstrap version of the quantile function obtained through the
bootstrap empirical distribution $\hat{F}_{ Y_{t-1}|X=x, D_{t}=1}^{\ast }$.
The bootstrap version of the CQTT process given $X_{i}=x$ is given by 
\begin{eqnarray*}
\hat{\Delta}_{x}^{QTT \ast}(\tau) := \hat{F}_{Y_{t}(1)|X=x, D_{t}=1}^{\ast
-1}(\tau) - \hat{F}_{Y_{t}(0)|X=x, D_{t}=1}^{\ast -1}(\tau),
\end{eqnarray*}
for $\tau \in \mathcal{T}$, where $\hat{F}_{Y_{t}(j)|X=x,D_{t}=1}^{\ast
-1}(\tau)$ is the $\tau$th conditional quantile of the bootstrap empirical distribution $%
\hat{F}_{ Y_{t}(j)|X=x,D_{t}=1}^{\ast}$ of potential outcomes for $j = 0,1$.

For the validity of the resampling method explained above, we need to
introduce the notion of conditional weak convergence in probability,
following \cite{VW1996}. For some normed space $\mathbb{D}$, let $BL_{1}(%
\mathbb{D})$ denote the space of all Lipschitz continuous functions from $%
\mathbb{D}$ to $[-1,1]$. Given the original sample $\mathbf{D}_{n}$ with $n$
being the sample size, consider a random element $B_{n}^{\ast}:=g(\mathbf{D}%
_{n}, \mathbf{W}_{n})$ as a function of the original sample and the random
weight vector $\mathbf{W}_{n}$ generating the bootstrap draw. The bootstrap
law of $B_{n}^{\ast}$ is said to consistently estimate the law of some tight
random element $B$, or $B_{n}^{\ast} \rightsquigarrow^{p} B$ if 
\begin{eqnarray*}
\sup_{h \in BL_{1}(\mathbb{D})} \big | E_{\mathbf{W}_{n}}[h(B_{n}^{\ast})] -
E[h(B)] \big | \to^p 0,
\end{eqnarray*}
where $E_{\mathbf{W}_{n}}$ is the conditional expectation with respect to $%
\mathbf{W}_{n}$ given the original sample $\mathbf{D}_{n}$.

To state a preliminary result, we define empirical processes indexed by $%
\mathcal{Y}_{s|x,d}$ as 
\begin{eqnarray*}
\hat{G}_{s, x}^{(d)\ast} := \sqrt{n} \big ( \hat{F}_{Y_{s}|X=x,
D_{t}=d}^{\ast} - \hat{F}_{Y_{s}|X=x, D_{t}=d} \big ),
\end{eqnarray*}
for $(s, d) \in \{t-1, t\} \times \{0, 1\}$. Also, define an empirical process indexed
by $\mathcal{Y}_{t|x, 1}(0)$ as 
\begin{eqnarray*}
\tilde{G}_{t, x}^{(0) \ast} := \sqrt{n} \big ( \tilde{F}_{Y_{t}(0)|X=x,
D_{t}=1}^{\ast} - \tilde{F}_{Y_{t}(0)|X=x, D_{t}=1} \big ),
\end{eqnarray*}
where the empirical distribution is given by 
\begin{eqnarray*}
\tilde{F}_{Y_{t}(0)|X=x, D_{t}=1}^{\ast}(y) := \frac{ 1 }{ n_{x}^{(0)} }
\sum_{i=1}^{n} W_{i}^{(0)} 1 \big \{ \Delta Y_{it} + F_{ Y_{t-1}|X=x,
D_{t}=1}^{-1} \circ F_{Y_{t-1}|X=x, D_{t}=0} ( Y_{i,t-1} ) \le y \big \} %
\delta_{i, x}^{0}.
\end{eqnarray*}
The following lemma shows that a set of the empirical processes defined
above consistently estimates the tight random element defined in Lemma \ref%
{lemma:asym-basic}.

\vspace{0.5cm}

\begin{lemma}
\label{lemma:asym-basic-B}  Suppose that Assumption A1-A6 and B hold. Then,
for each $x \in \mathcal{X}$,  
\begin{eqnarray*}
\big ( \tilde{G}_{t, x}^{(0)\ast}, \hat{G}_{t-1, x}^{(0)\ast}, \hat{G}_{t,
x}^{(1)\ast}, \hat{G}_{t-1, x}^{(1)\ast} \big ) \rightsquigarrow^{p} \big ( 
\mathbb{V}_{x}^{(0)}, \mathbb{W}_{x}^{(0)}, \mathbb{V}_{x}^{(1)}, \mathbb{W}%
_{x}^{(1)} \big ),
\end{eqnarray*}
in  $\mathbb{S}_{x}$,  where the limit processes defined in  Lemma \ref%
{lemma:asym-basic}.
\end{lemma}

\vspace{0.5cm} 

Using this lemma, we first show that the exchangeable bootstrap provides a
way to consistently estimate limit process of a pair of empirical
distributions of potential outcomes. Subsequently we argue that the limit
process of the CQTT estimator can be estimated, using the functional delta
method for a Hadmard differentiable map. The result is summarized in the
following theorem.

\vspace{0.5cm}

\begin{theorem}
\label{theorem:F-convergence-B}  Define  $ \hat{Z}_{x}^{(j)\ast}(y)  :=  
\sqrt{n}  \big (
\hat{F}_{Y_{t}(j)|X=x, D_{t}=1}^{\ast}(y)  -  \hat{F}_{Y_{t}(j)|X=x,
D_{t}=1}(y)  \big )
$  
for $j \in \{0,1\}$, $x \in \mathcal{X}$  and  $y \in \mathcal{Y}%
_{t|x,1}(j)$.  Suppose that Assumption A1-A6 and B hold. Then,  
for each $x \in \mathcal{X}$,  
\begin{eqnarray*}
  (\hat{Z}_{x}^{(0)\ast}, \hat{Z}_{x}^{(1)\ast})  \rightsquigarrow^{p}  
  (\mathbb{Z}_{x}^{(0)},\mathbb{Z}_{x}^{(1)}),  
\end{eqnarray*}
and thus
the exchangeable bootstrap procedure consistently estimates  the law of the
limit stochastic process of the CQTT:  
\begin{eqnarray*}
\sqrt{n} \big ( \hat{\Delta}_{x}^{QTT \ast}(\tau) - \hat{\Delta}%
_{x}^{QTT}(\tau) \big ) \rightsquigarrow^{p} \bar{\mathbb{Z}}%
_{x}^{(1)}(\tau) - \bar{\mathbb{Z}}_{x}^{(0)}(\tau), \ \ \ \ \tau \in 
\mathcal{T}.
\end{eqnarray*}
\end{theorem}

\vspace{0.5cm} 


\section{Monte-Carlo Simulation}

We consider a small scale Monte Carlo simulation to assess the performance
of our estimator in finite samples and consider the effect of small
deviations from the Copula Invariance assumption. The data generating
process (DGP) for potential outcomes is given by 
\begin{align*}
Y_{it}(d) = \mu(d) + \theta_t + v_i + \epsilon_{it}
\end{align*}
In this setup, the treatment effect is constant across all quantiles and
given by $\mu(1) - \mu(0)$. $\theta_t$ is a time fixed effect that is common
across individuals; $v_i$ is time invariant unobserved heterogeneity that
can be distributed differently across treated and untreated groups; and $%
\epsilon_{it}$ are time varying unobservables. Throughout, we impose that $%
\theta_t=1$ and set $\mu(1) - \mu(0)$ to be either 1 or 0 and label this
effect TE.

\paragraph{DGP 1:}

The first DGP imposes that $v_i | D=d \sim N(d,1)$ and that $\epsilon_{it}$
is a noise term that follows a standard normal distribution. One can show
that both our model and the Change in Changes model \citep{athey-imbens-2006}
hold under this setup. We use this DGP to assess the finite sample
performance of our estimator using the Change in Changes model as a
benchmark. We perform 1000 Monte Carlo simulations and at each iteration, we
use 1000 bootstrap iterations to calculate empirical rejection frequencies
given the nominal size of 5\%. We
calculate standard errors using the empirical block bootstrap for the Change
in Changes using the same sample and 1000 bootstrap iterations.

The results are presented in Table 1. The first panel of Table 1 considers
the case where the treatment effect is 0 at all quantiles. Relative to the
Change in Changes models, our estimator is less biased in finite samples
especially at the 0.9th quantile. With only 100 observations, our inference procedure is
somewhat undersized, but it exhibits good size properties with 200 or 500
observations. Finally, the second panel considers the case where the
treatment effect is 1 at all quantiles. The power of our inference procedure increases
rapidly as the sample size increases from 100 to 200 and then to 500. It
also has more power at the median than at the 0.1th or 0.9th quantiles.

\paragraph{DGP 2:}

For the second DGP, we want to assess the effect of small deviations from
the Copula Invariance assumption while the Distributional DID assumption continues to hold. To do this, we assume that 
\begin{align*}
(v_i, \epsilon_{i2}, \epsilon_{i1}) | D=d \sim N(0,V_d)
\end{align*}
where 
\begin{align*}
V_d = 
\begin{pmatrix}
1 & \rho_{dv2} & \rho_{dv1} \\ 
\rho_{dv2} & 1 & \rho_{d12} \\ 
\rho_{dv1} & \rho_{d12} & 1%
\end{pmatrix}%
\end{align*}

Under this setup, $(Y_{i1}(0),\Delta Y_{i2}(0)|D=d)$ has a bivariate normal
distribution with correlation parameter $\rho _{d2}-\rho _{d1}+\rho _{d12}-1$%
. For bivariate normal distributions, the copula is Gaussian with the
dependence parameter given by the correlation coefficient. For DGP 2, we set 
$\rho _{d1}=0$ and $\rho _{d12}=1/2$ both for $d=0,1$; then, we set $\rho
_{d2}=d\bar{\rho}$ and vary $\bar{\rho}$. For $\bar{\rho}=0$, the Copula
Invariance assumption holds, but it is violated when $\bar{\rho}\neq 0$. For
each simulation, we consider the case with $N=200$.

The results are presented in Table 2. Small violations of the Copula
Invariance assumption ($\bar{\rho}=0.05$ or $\bar{\rho}=0.10$) lead to small
increases in the bias of our estimator. For example, for the 0.1th quantile,
the bias increases from 0.020 to 0.073 and 0.121 as $\bar{\rho}$ increases
from 0.00 to 0.05 to 0.10. A large increase in the violation of the 
Copula Invariance
assumption, $\bar{\rho}=0.50$, leads to a much larger increase in the bias
of our estimator for the 0.1th quantile (bias increases to 0.425) in our
simulations. On the other hand, the results for the 0.5th quantile are almost
completely insensitive to deviations from the Copula Invariance assumption. In the large
deviation case, $\bar{\rho}=0.5$, the bias is very small: 0.013; nor does
the root mean squared error change much even with large violations of the Copula Invariance
assumption for the 0.5th quantile.


\section{Empirical Application}

To illustrate our method, we consider the effect of increases in the minimum
wage on the distribution of earnings. In 1996, President Bill Clinton signed
a bill to increase the federal minimum wage from \$4.25 to \$5.15 per hour
by September 1997. The minimum wage did not increase again until the Fair
Minimum Wage Act of 2007. The Fair Minimum Wage Act was proposed on January
5, 2007; signed on May 25, 2007 by President George W. Bush; increased the
minimum wage to \$5.85 on July 24, 2007; and in two more gradual increases
settled on \$7.25 per hour in July 2009.

We exploit the long period from 1999-2007 with a flat, national minimum wage
to use state-level variation in the minimum wage to identify and estimate
the effect of increasing the minimum wage on the earnings distribution. In
the first quarter of 2006, for 33 states the federal minimum wage was
binding. The other states had state minimum wages that were higher than the
federal minimum wage. For our analysis, we take a subset of states that
raised their minimum wage in the first quarter of 2007 and have a close
geographic proximity to a state whose effective minimum wage is given by the
federal minimum wage for the entire period. This results in a sample of 5
states that increase their minimum wage (Arizona, Colorado, Minnesota,
Missouri, and North Carolina) -- this is the treated group -- and 14 control
states (Georgia, Idaho, Iowa, Kansas, Kentucky, Nebraska, New Mexico, North
Dakota, South Carolina, South Dakota, Tennessee, Utah, Virginia, and
Wyoming). The large literature on estimating the effect of changes in
minimum wage policies has emphasized that unconditional DID methods are not
likely to be valid \citep{dube-lester-reich-2010}. There has also been an
interest in understanding the effect of minimum wages on the distribution of
earnings \citep{dube-2013}.

The data for the application comes from the Current Population Survey (CPS) %
\citep{ipums-cps-2015}. The CPS surveys roughly 140,000 individuals per
month. Individuals are interviewed for four consecutive months, out of the
sample for eight months, and then interviewed for four more months.
Importantly for our purposes, individual earnings questions are asked in the
4th month and in the 8th month in the sample -- due to the survey design,
these are exactly one year apart. There are some difficulties with linking
the CPS over time, but longitudinal identifiers are available in the IPUMS
database \citep{drew-flood-warren-2014}. We limit the sample to individuals
who have earnings greater than \$10 per week and to those that we can
successfully link over time. This procedure results in a sample size of 8256
individuals (2 observations per individual) that are observed in the first
two quarters of 2007 and at some point in 2006.

Next, we divide the data into 8 categories based on gender, race (white or
non-white), and education attainment (college graduate or not). Summary
statistics are provided in Table 3. There are several important differences
between treated states and untreated states. First, in 2006, earnings in
states that raised their minimum wage were 6.5 log points (statistically
significant) higher, on average, than in states that did not raise their
minimum wage. This provides some evidence that cross sectional comparisons
of earnings distributions, at least without adjusting for covariates, are
likely to lead to upwardly biased estimates of the effect of the minimum
wage on the earnings distribution. Second, there are differences in the
covariates across treated and untreated states. While the fraction of male
individuals is similar across treated and untreated states, individuals in
states that raised their minimum wage are more likely to be white and more
likely to have a college degree. If the path of earnings, in the absence of
changes in minimum wage policy, depends on race and education, then it will
be important to control for these covariates in the analysis; similarly, if
the dependence between the change in outcomes over time and the initial
level of outcomes (both in the absence of changes in minimum wage policy)
depends on these variables, then it will be important to condition on these
variables. Our method makes it possible to account for both of these
complications.

To evaluate the effect of minimum wage increases across each subgroup, first we test whether the CQTT is 0 across all quantiles using the Kolmogorov-Smirnov test statistic.  In order to do this, we estimate the CQTT for each group over a fine grid of $\tau$ from 0.05 to 0.95 by 0.01 and using 1000 boostrap iterations to calculate the critical values of the test.  These results are available in Table 4.  For only 3 out
of 8 of the race-gender-education subgroups are we able to reject the null
that the distribution of earnings is the same due to the change in the
minimum wage policy. We reject the null of no effect at any quantile for (i)
white, female, college graduates; (ii) non-white, male, non-college, and
(iii) non-white, female, non-college. Interestingly, the groups for which we
can reject the null tend to have lower earnings than other groups -- this
seems intuitive as the minimum wage is binding only at the lower part of the
earnings distribution.  Table 4 also provides estimates of the CQTT for each subgroup at the 0.1, 0.5, and 0.9 quantiles with pointwise standard errors computed with 1000 bootstrap iterations reported.

Second, Figure 1 plots the CQTT for each subgroup as well as
95\% confidence bands. The confidence bands are obtained by inverting the Kolmogrov-Smirnov test mentioned above.  For each of the groups that we reject the null of
no effect, we only find statistically significant results at the lower part
of the distribution which also corresponds to our intuition about the
effects of increasing the minimum wage. Interestingly, the effect of the
minimum wage on earnings appears to be negative.  This result may appear surprising, but it should be remembered that earnings mixes both wages and hours.  Thus, even if the minimum wage lifts wages in the lower part of the distribution, our result could be be explained by a decrease in hours due to minimum wage increases.%
\footnote{%
As a robustness check, we also compare the change in the 10th percentile of
earnings between 2007 and 2006 for the treated group to the same change for
the untreated group. For the treated group, the 10th percentile of log
earnings increased by 3.0 log points, but the 10th percentile for the
untreated group increased by 10.5 log points. This result is in line with
our results that say that the distribution of earnings for subgroups that
experienced an effect of the minimum wage tended to be worse than the
distribution of earnings would have been absent changes in the minimum wage
policy.}


\section{Conclusion}

This paper has considered identifying and estimating the Conditional
Quantile Treatment Effect on the Treated under a Distributional DID
assumption when only two periods of data are available. We have developed
uniform confidence intervals for the CQTT and shown the validity of a
bootstrap procedure for computing confidence bands. Finally, we estimated
conditional quantile treatment effects for states that increased their
minimum wage.

Methodologically, the key innovation is to recover the unknown dependence
between the change and initial level of untreated potential outcomes for the
treated group from the observed dependence from the untreated group.
Combining this condition with a distributional extension of the most common
mean DID assumption results in point identification of the counterfactual
conditional distribution of untreated potential potential outcomes for the
treated group; and, therefore, to identification of the CQTT. There are many
examples in finance, auction models, and duration models where
identification depends on an unknown copula. The idea of replacing an
unknown copula with one observed for another group may prove to be a
fruitful line of research in those cases.

\clearpage  
\setstretch{0.3} 
\bibliographystyle{ecta.bst}
\bibliography{REF.bib}

\newpage

\section*{Appendix}


\renewcommand{\thelemma}{A\arabic{lemma}} \renewcommand{\theproposition}{A%
\arabic{proposition}} \renewcommand{\theequation}{A\arabic{equation}} %
\setcounter{equation}{0} \setcounter{lemma}{0}

In Appendix, we use $\| \cdot \|$ to denote the Euclidean norm for vectors.



\vspace{0.5cm} 
\begin{proof}
  [\textbf{Proof of Theorem \ref{theorem:identification-1}}]
  Let $x \in \mathcal{X}$ be fixed.
  For every $y \in \supp(Y_{it}(0)|X_{i}=x, D_{it}=1)$,
  we can write that 
  $F_{Y_{t}(0) |X=x, D_{t}=1}(y) 
   = 
   \Pr \{ 
     \Delta Y_{it}(0) + Y_{i, t-1}(0) \le y| X_{i}=x, D_{it}=1
   \}
  $.
  Define 
  \begin{eqnarray}
    \label{eq:d-transform}
    U_{i}^{d}:= F_{\Delta Y_{t}(0)| X=x, D_{t} = d}
    \big (\Delta Y_{it}(0) \big)
    \  \ \ \mathrm{and} \ \ \
    V_{i}^{d}:= F_{Y_{t-1}(0)| X=x, D_{t} = d}
    \big(Y_{i,t-1}(0) \big),
  \end{eqnarray}
  for $d \in \{ 0, 1 \}$.
  Under Assumption A4, we have
  \begin{eqnarray}
    \label{eq:q-transform}
    \Delta Y_{it}(0) = F_{\Delta Y_{t}(0)|X=x, D_{t} = d}^{-1}(U_{i}^{d})
    \ \ \mathrm{and} \ \  
    Y_{i, t-1}(0) = F_{Y_{t-1}(0)|X=x, D_{t} = d}^{-1}(V_{i}^{d}),
  \end{eqnarray}
  almost surely
  \citep[see][]{Rosenblatt1952AMS}.
  It follows that  
  \begin{eqnarray*}
    F_{Y_{t}(0)|X=x, D_{t}=1}(y)
    =
    \Pr
    \big \{
    F_{\Delta Y_{t}(0)|X=x, D_{t}=1}^{-1}(U_{i}^{1})
    +
    F_{Y_{t-1}(0)|X=x, D_{t}=1}^{-1}(V_{i}^{1}) 
    \le y
    \big |X_{i}=x, D_{it} = 1
    \big \}.
  \end{eqnarray*}
  For each
  $d \in \{0,1 \}$,
  the joint distribution of 
  $(U_{i}^{d},V_{i}^{d})$
  conditional on $(X_{i}, D_{it}) =(x, d)$
  is given by 
  a conditional copula 
  $C_{\Delta Y_{t}(0), Y_{t-1}(0) |X=x, D_{t}=d}$,
  which 
  is invariant with respect to 
  the conditional variable $D_{it}$
  under Assumption A3. 
  Thus we have
  \begin{eqnarray*}
    F_{Y_{t}(0)|X=x, D_{t}=1}(y)
    = 
    \Pr
    \big \{
    F_{\Delta Y_{t}(0)|X=x, D_{t}=1}^{-1}(U_{i}^{0})
    +
    F_{Y_{t-1}(0)|X=x, D_{t}=1}^{-1}(V_{i}^{0}) 
    \le y
    \big | X_{i}=x, D_{it} = 0
    \big \}.
  \end{eqnarray*}
  Under Assumption A2,
  $
  F_{\Delta Y_{t}(0)|X=x, D_{t}=1}^{-1}(\cdot)
  =
  F_{\Delta Y_{t}(0)|X=x, D_{t}=0}^{-1}(\cdot)
  $,
  which with 
  (\ref{eq:q-transform})
  yields 
  that 
  $  F_{\Delta Y_{t}(0)|X=x, D_{t}=1}^{-1}(U_{i}^{0})
    = 
    \Delta Y_{it}(0)
  $,
  almost surely. 
  Also,
  using the relation in (\ref{eq:d-transform})
  we can show that
  $ 
    F_{Y_{t-1}(0)|X=x, D_{t}=1}^{-1}(V_{i}^{0}) 
    =
    F_{Y_{t-1}(0)|X=x, D_{t}=1}^{-1}
    \circ 
    F_{Y_{t-1}(0)|X=x, D_{t} = 0}(Y_{i,t-1}(0)),
  $  
  almost surely. 
  Hence, the desired result follows.
\end{proof} 


\vspace{0.5cm} 
\begin{proof}
  [\textbf{Proof of Corollary \ref{corollary:identification-EX}}]
  Let $x \in \mathcal{X}$ be fixed.  
  Given that 
  the data generating process satisfies Assumption A1-A4,
  the result in Theorem \ref{theorem:identification-1} holds
  and we have 
  \begin{eqnarray*}
     F_{Y_{t}(0)| X=x, D_{t}=1}
    (y)
    =
    \Pr
    \big \{
      \Delta Y_{it}(0) 
      + 
      F_{ Y_{t-1}|X=x,D_{t}=1}^{-1}
      \circ 
      F_{Y_{t-1}|X=x,D_{t}=0}
      ( Y_{i,t-1} ) 
    \le y 
      | X_i=x,
      D_{it} = 0
    \big \}, 
  \end{eqnarray*}
  for $y \in \supp(Y_{it}(0)| X_i=x, D_{it}=1)$.
  Because of the repeated cross section,  
  we cannot identify the term 
  $\Delta Y_{it}(0):= Y_{it}(0) - Y_{i,t-1}(0)$
  from the observed outcomes of the untreated group.
  Under the rank invariance assumption, however, we have 
  \begin{eqnarray*}
    F_{Y_{t}(0)|X=x,D_{t} =0}(Y_{it}(0))
    =
    F_{Y_{t-1}(0)|X=x,D_{t} = 0}(Y_{i,t-1}(0)), 
  \end{eqnarray*}
  where 
  the distributions 
  $F_{Y_{t}(0)|X=x,D_{t} =0}$
  and 
  $F_{Y_{t-1}(0)|X=x,D_{t} = 0}$
  of potential outcomes,
  can be identified 
  by the distributions 
  $F_{Y_{t}|X=x,D_{t} =0}$
  and 
  $F_{Y_{t-1}|X=x,D_{t} = 0}$
  of observed outcomes,
  respectively. Thus,
  we can identify $\Delta Y_{it}(0)$ for individuals with $X_{i}=x$ and $D_{it}=0$
  by 
  \begin{eqnarray*}
    \widetilde{\Delta Y_{it}}(0) 
    :=
    F_{Y_{t}|X=x,D_{t} =0}^{-1} \circ 
    F_{Y_{t-1}|X=x,D_{t} = 0}(Y_{i,t-1})
    - 
    Y_{i,t-1}. 
  \end{eqnarray*}
  This leads to the desired result.
\end{proof}
\vspace{0.5cm} 


To derive the limiting distribution of the estimator for the CQTT, we
present two technical lemmas concerning the Hadamard differentiability. We
introduce a setup and notations used in the these lemmas. Let $%
F_{0}:=(G_{0}, H_{0})$ with $G_{0}$ and $H_{0}$ being distribution functions
having a compact support $\mathcal{V} \subset \mathbb{R}$ and a density
function $g_{0}$ and $h_{0}$, respectively. Consider a pair of continuous
random variables $(V_{1}, V_{2})$ taking values on $\mathcal{V} {\times} 
\mathcal{V} $ with the joint distribution $F_{V_{1} V_{2}}$ having a density 
$f_{V_{1} V_{2}}$ as well as the marginal distributions $F_{V_{j}}$ having a
density $f_{V_{j}}$ for $j=0,1$. We suppose that the conditional
distribution $F_{V_{1}|V_{2}}$ has a continuous density function $%
f_{V_{1}|V_{2}}$ uniformly bounded away from 0 and $\infty$.

\vspace{0.5cm}

\begin{lemma}
\label{lemma:H-diff-1}  Let  $\mathbb{D}  :=  ( C(\mathcal{V}) )^2  $  and 
define the map  $\psi:  \mathbb{D}_{\psi}  \subset  \mathbb{D}  \mapsto 
\ell^{\infty}(\mathcal{V})  $,  given by  
\begin{eqnarray*}
\psi(F) := G^{-1} \circ H,
\end{eqnarray*}
for  $F:=(G, H) \in \mathbb{D}_{\psi}$,  where  $\mathbb{D}_{\psi}:=  
\mathbb{E}{\times}\mathbb{E}  $  with  $\mathbb{E}$  denoting  the set of 
all distributions functions having a strictly positive, bounded density. 
Then,  the map $\psi$  is Hadamard differentiable at  $F_{0}$  tangentially
to  $\mathbb{D}$.  Its derivative at $F_{0}$ in $\gamma:=(\gamma_{1},
\gamma_{2}) \in \mathbb{D}$  is given by  
\begin{eqnarray*}
\psi_{F_{0}}^{\prime }(\gamma) = \frac{ \gamma_{2} - \gamma_{1} \circ
G_{0}^{-1} \circ H_{0} }{ g_{0} \circ G_{0}^{-1} \circ H_{0} }.
\end{eqnarray*}
\end{lemma}

\begin{proof}[\textbf{Proof}]
  To prove the assertion, we first 
  represent 
  $\psi$ as a composition map.
  Let 
  $\mathbb{D}_{\psi_{2}}
    :=
    \mathbb{E}^{-}
    {\times} 
    C(\mathcal{V})
  $,
  where 
  $\mathbb{E}^{-}$
  denotes 
  the set of 
  generalized inverse
  of 
  all functions in $\mathbb{E}$.
  Define the maps 
  $\psi_{1}: 
  \mathbb{D}_{\psi}
  \mapsto 
  \mathbb{D}_{\psi_{2}}$
  and 
  $\psi_{2}: 
    \mathbb{D}_{\psi_{2}}
    \mapsto 
    \ell^{\infty}(\mathcal{V})
  $,
  given by 
  \begin{eqnarray*}
    \psi_{1}(\Gamma):=(\Gamma_{1}^{-1}, \Gamma_{2})  
    \ \ \ \ \mathrm{and} \ \ \ \ 
    \psi_{2}(\Lambda):= \Lambda_{1} \circ \Lambda_{2},
  \end{eqnarray*}
  for 
  $  \Gamma:=(\Gamma_{1}, \Gamma_{2}) \in \mathbb{D}_{\psi}$
  and 
  $\Lambda:=(\Lambda_{1}, \Lambda_{2}) \in \mathbb{D}_{\psi_{2}}$.
  Then 
  we can write 
  $\psi = \psi_{2} \circ \psi_{1} $.

  For the map $\psi_{1}$,
  Lemma 3.9.23(ii) of \cite{VW1996} implies that if $\Gamma$ 
  has a derivative denoted by $\Gamma'$, 
  then the map $\psi_{1}$ is Hadamard differentiable at $\Gamma$
  tangentially to 
  $\mathbb{D}$.
  Its derivative at $\Gamma$
  in 
  $\gamma:=(\gamma_{1}, \gamma_{2}) 
   \in \mathbb{D}
  $ 
  is given by
  \begin{eqnarray*}
    \psi_{1, \Gamma}'(\gamma)
    :=
    \big (
      -(\gamma_{1}/ \Gamma_{1}') \circ \Gamma_{1}^{-1}, \gamma_{2}
    \big ).
  \end{eqnarray*}
  In terms of the map $\psi_{2}$,
  Lemma 3.9.27 of \cite{VW1996} implies that 
  $\psi_{2}$
  is Hadamard differentiable at $\Lambda$
  tangentially to  
  $
   C([0,1])
   {\times}
   \ell^{\infty}(\mathcal{V})
  $. 
  Its derivative at $\Lambda$
  in 
  $\lambda:=(\lambda_{1}, \lambda_{2}) 
   \in 
   C([0,1])
   {\times}
   \ell^{\infty}(\mathcal{V})
  $
  is given by 
  \begin{eqnarray*}
    \psi_{2, \Lambda}'(\lambda)
    := 
    \lambda_{1} \circ \Lambda_{2}
    +
    \Lambda_{1, \Lambda_{2}}' \lambda_{2}.
  \end{eqnarray*}

  Lemma 3.9.3 of \cite{VW1996} 
  with Hadamard derivatives of the maps 
  $\psi_{1}$ and $\psi_{2}$
  yields
  that 
  $\psi_{F_{0}}'(\gamma) = 
    \psi_{2, (G_{0}^{-1}, H_{0})}' 
    \circ 
    \psi_{1, F_{0}}'(\gamma)
  $
  for 
  $\gamma \in \mathbb{D}$,
  where 
  \begin{eqnarray*}
    \psi_{1, F_{0}}'(\gamma)
    = 
    \big (
    - 
    ( \gamma_{1}/ g_{0} )
    \circ G_{0}^{-1}, 
    \gamma_{2}
    \big ),  
  \end{eqnarray*}
  and 
  \begin{eqnarray*}
    \psi_{2, (G_{0}^{-1}, H_{0})}'(\lambda)
    = 
    \lambda_{1} \circ H_{0}
    +
    \frac{
      \lambda_{2} 
    }{
      g_{0} \circ G_{0}^{-1} \circ H_{0}
    }, 
  \end{eqnarray*}
  because 
  $\partial G_{0}^{-1}(\tau) / \partial \tau
   =
   1 / 
   \big ( g_{0} \circ G_{0}^{-1}(\tau) \big ).
  $
  Hence the desired result follows.
\end{proof}
\vspace{0.5cm} 


\begin{lemma}
\label{lemma:H-diff-2}  Let  $\mathbb{D}  :=  ( C(\mathcal{V}) )^2  $  and 
let  $\mathcal{W}$  be a compact subset of $\mathbb{R}$.  Define the map  $%
\phi:  \mathbb{D}_{\phi}  \subset  \mathbb{D}  \mapsto  \ell^{\infty}(%
\mathcal{W})  $,  given by  
\begin{eqnarray*}
\phi(F)(w):= \Pr\{V_{1} + G^{-1} \circ H(V_{2}) \le w\},
\end{eqnarray*}
for  $F:=(G, H) \in \mathbb{D}_{\phi}$  and  for  $w \in \mathcal{W}$, 
where  $\mathbb{D}_{\phi}:=  \mathbb{E}{\times}\mathbb{E}  $  with  $\mathbb{%
E}$  being  the set of  all distributions functions having a strictly
positive, bounded density.  Then,  the map $\phi$  is Hadamard
differentiable at  $F_{0}$  tangentially to $\mathbb{D}$.  Its derivative at 
$F_{0}$  in  $\gamma:=(\gamma_{1}, \gamma_{2}) \in \mathbb{D}$  is given by  
\begin{eqnarray*}
\phi_{F_{0}}^{\prime }(\gamma)(w) := \int \big ( \gamma_{2} (v_{2}) -
\gamma_{1} \circ G_{0}^{-1} \circ H_{0} (v_{2}) \big ) \frac{ f_{V_{1}
V_{2}} ( w - G_{0}^{-1} \circ H_{0} (v_{2}), v_{2} ) }{ g_{0} \circ
G_{0}^{-1} \circ H_{0} (v_{2}) } d v_{2}.
\end{eqnarray*}
\end{lemma}

\begin{proof}[\textbf{Proof}]
  To prove the assertion, we represent 
  $\phi$ as a composition map.
  Define $\psi: \mathbb{D}_{\phi} \to \mathbb{D}_{\pi}$
  as in the proceeding lemma,
  where 
  $\mathbb{D}_{\pi}$
  denotes the set of all functions 
  $F^{-1} \circ G$
  for 
  $(F^{-1}, G) \in 
   \mathbb{E}^{-}
   {\times}
   \mathbb{E}
  $
  with 
  $\mathbb{E}^{-}$
  and
  $\mathbb{E}$
  defined in the proof of the proceeding lemma.
  Define 
  the map
  $\pi:
    \mathbb{D}_{\pi}
    \mapsto 
    \ell^{\infty}(\mathcal{W})
  $,
  given by 
  \begin{eqnarray*}
    \pi(\Xi)(w):= 
    \int F_{V_{1}|V_{2}} 
    \big (
    w - \Xi(v_{2}) | v_{2}
    \big)
    d F_{V_{2}}(v_{2}),
  \end{eqnarray*}
  for $w \in \mathcal{W}$.
  Since we can write 
  $
    \phi(F)(w)
    = 
    \int 
    F_{V_{1}|V_{2}} 
    \big (
       w - G^{-1} \circ H(v_{2})  | v_{2}
    \big)
    d F_{V_{2}}(v_{2})
  $
  for 
  $F \in \mathbb{D}$
  and 
  $w \in \mathcal{W}$,
  we can show that 
  $\phi = \pi \circ \psi$.

  We wish to show that $\pi$ has a Hadamard derivative 
  at $\Xi \in \mathbb{D}_{\pi}$
  tangentially to $\mathbb{D}$
  with derivative at $\Xi$ in $\xi \in \mathbb{D}$
  \begin{eqnarray}
    \label{eq:hd-pi}
    \pi_{\Xi}'(\xi)(w)
    = 
    \int 
    \xi (v_{2})
    f_{V_{1}|V_{2}}(w - \Xi(v_{2})|v_{2} ) 
    d F_{V_{2}}(v_{2}). 
  \end{eqnarray}
  Consider any sequence $t_{k} > 0$ and $\Xi_{k} \in \mathbb{D}_{\pi}$
  for $k \in \mathbb{N}$ such that 
  $t_{k} \searrow 0$
  and 
  $\xi_{k}:= (\Xi_{k} - \Xi)/t_{k} \to \xi$
  in 
  $\mathbb{D}$
  as $k \to \infty$.
  We have 
  \begin{eqnarray*}
    F_{V_{1}|V_{2}}(w - \Xi_{k}(v_{2}) |v_{2})
    -
    F_{V_{1}|V_{2}}(w - \Xi(v_{2})|v_{2})
    = 
    t_{k} \xi_{k} (v_{2})
    \int_{0}^{1}
    f_{V_{1}|V_{2}}(w - \Xi(v_{2}) - r t_{k} \xi_{k}(v_{2})|v_{2})
    dr.
  \end{eqnarray*}
  It follows that 
  \begin{eqnarray*}
    \frac{
      \pi(\Xi_{k})
      -
      \pi(\Xi)
    }{ t_{k}}
    -
    \pi_{\Xi}'(\xi)
    &=& 
    \int 
    \big (
      \xi_{k} (v_{2})
      -
      \xi (v_{2})
    \big )
    f_{V_{1}|V_{2}}(\cdot - \Xi(v_{2})|v_{2} )
    d F_{V_{2}}(v_{2}) \\
    &&+
    \int 
    \xi_{k} (v_{2})
    D_{k}(\cdot, v_{2})
    d F_{v_{2}}(v_{2}),
  \end{eqnarray*}
  where 
  $ D_{k}(w, v_{2}):=
    \int_{0}^{1}
    \big \{
    f_{V_{1}|V_{2}}(w - \Xi(v_{2}) - r t_{k} \xi_{k}(v_{2})|v_{2})
    -
    f_{V_{1}|V_{2}}(w - \Xi(v_{2}) |v_{2})
    \big \}
    dr
  $. 
  Since 
  $f_{V_{1}|V_{2}}$ is uniformly continuous and 
  $\xi_{k}$ is uniformly bounded, 
  $\lim_{k \to \infty}\| D_{k}\|_{\mathcal{W}{\times}\mathcal{V}} = 0$
  and thus 
  the second term on the above display converges to 0 
  as $k \to \infty$.
  Also 
  the first term on the above display converges to zero 
  because $f_{V_{2}|V_{1}}$ is uniformly bounded
  and 
  $ 
    \| 
      \xi_{k} 
      -
      \xi
    \|_{\infty} \to 0
  $  as $k \to \infty$.
  Thus the map $\pi$ has the Hadamard derivative as stated.

  Lemma 3.9.3 of \cite{VW1996} shows that 
  $\phi_{F_{0}}'(\gamma) = 
    \pi_{G_{0}^{-1} \circ H_{0} }' 
    \circ 
    \psi_{F_{0}}'(h)
  $,
  which 
  together with
  the Hadamard derivative
  of 
  $\pi$
  in
  (\ref{eq:hd-pi})
  and 
  the one of 
  $\psi$
  in 
  Lemma \ref{lemma:H-diff-1}
  yields
  \begin{eqnarray*}
    \phi_{F_{0}}'(\gamma)(w)    
    =
    \int 
    \frac{
      \gamma_{2}    (v_{2})
      - 
      \gamma_{1} \circ G_{0}^{-1} \circ H_{0}  (v_{2})
    }{
      g_{0} \circ G_{0}^{-1} \circ H_{0}     (v_{2})
    }
    f_{V_{1}|V_{2}}
    \big (
      w - G_{0}^{-1}\circ H_{0}(v_{2})
      \big |v_{2} 
    \big )
    d F_{V_{2}}(v_{2}).
  \end{eqnarray*}
  Hence the desired result follows.
\end{proof}
\vspace{0.5cm} 



Define $\tilde{V}  :=  V_{1} + G_{0}^{-1}\circ H_{0}(V_{2})$.
We additionally assume that $\tilde{V}$ is distributed over a compact space $%
\mathcal{V}$ with a distribution $F_{\tilde{V}}$ and a continuous density $f_{%
\tilde{V}}$ uniformly bounded away from 0 and $\infty$. We consider random
sample $\{(V_{1 i}, V_{2 i})\}_{i=1}^{n}$ of $n$ independent copies of $%
(V_{1}, V_{2})$ and let $\tilde{V}_{i}  :=  V_{1i} + G_{0}^{-1}\circ H_{0}(V_{2i}) $. 
We set 
$F_{n}:=(G_{n}, H_{n})$ to denote a random element of $(\ell^{\infty}(%
\mathcal{V}))^2$ as a consistent estimator for $F_{0}$. 
For $F=(G, H)  \in  (C(\mathcal{V}))^2 
$ and $w \in \mathcal{W}$, define a functional taking values at $
\ell^{\infty}(\mathcal{W}) $: 
\begin{eqnarray}
  \label{eq:def-phi}
  \phi_{n} (F)(w):= n^{-1} \sum_{i=1}^{n} 1 \{ V_{1i} + G^{-1} \circ
H(V_{2i}) \le w \},
\end{eqnarray}
and the empirical process indexed by $F  \in  \big (\ell^{\infty}(\mathcal{V}) \big)^{2} $%
: 
\begin{eqnarray*}
  \nu_{n}(F) := \sqrt{n} \big ( \phi_{n}(F) - \phi(F) \big ).
\end{eqnarray*}
The lemma below is proven, along the line of Theorem 2.3 of \cite{VW2007IMS} with
some modification.

\vspace{0.5cm}

\begin{lemma}
\label{lemma:zero}  Suppose that  $\sqrt{n}(F_{n} - F_{0} )$  converges in
distribution to a tight, random element  with values in $(\ell^{\infty}(%
\mathcal{V}) )^2$.  Then,  
\begin{eqnarray*}
  \sup_{w \in \mathcal{W}} \big | \nu_{n} (F_{n}) - \nu_{n} (F_{0})\big |(w) = o_p(1).
\end{eqnarray*}
\end{lemma}

\begin{proof}[\textbf{Proof}]
  Because  the set $\mathcal{W}$ is compact, it suffices to shows 
  that the assertions holds for each $w \in {\cal W}$.
  Also, from the definition of $\phi_{n}$ in (\ref{eq:def-phi}), 
  both $\phi_{n}(F)$ and $\nu_{n}(F)$ can be considered as 
  functions of $\Xi:=G^{-1} \circ F \in \ell^{\infty}({\cal V})$.
  Letting  
  $
  \bar{\nu}_{n}(\Xi)= \nu_{n} (F)
  $,
  we show that, 
  for each $w \in \mathcal{W}$,
  \begin{eqnarray*}
    \big | \bar{\nu}_{n} (\Xi_{n}) - \bar{\nu}_{n} (\Xi_{0})\big |(w) = o_p(1),
  \end{eqnarray*}
  where 
  $
  \Xi_{n}:=G_{n}^{-1} \circ F_{n}
  $
  and 
  $\Xi_{0}:=G_{0}^{-1} \circ F_{0}$.
  Let $w \in \mathcal{W}$ be fixed and choose an 
  arbitrary small $\epsilon>0$.
  Suppose that 
  $\sqrt{n}(F_{n} - F_{0})$
  converges in distribution to some tight random element. 
  Then, by the functional delta method and 
  Hadamard differentiability, 
  Lemma \ref{lemma:H-diff-1} implies
  \begin{eqnarray}
    \label{eq:w-limit-1}
    \xi_{n}:=
    \sqrt{n}
    \big (
    \Xi_{n}
    -
    \Xi_{0}
    \big )
    \rightsquigarrow
    \xi_{\infty},
  \end{eqnarray} 
  in $\ell^{\infty}(\mathcal{V})$
  for some tight random element $\xi_{\infty}$.
  Then there exists a compact set 
  $S \subset \ell^{\infty}(\mathcal{V})$
  such that 
  $\Pr\{\xi_{\infty} \not \in S\} \le \epsilon/2$,
  and also  
  $\limsup_{n \to \infty} \Pr\{\xi_{n} \not \in S^{\delta/4}\} \le  \epsilon/2$
  for any $\delta>0$,
  where 
  $S^{\delta/4}$ is the 
  $\delta/4$-enlargement set of $S$. 
  Because $S$ is compact, 
  for any $\delta>0$
  there exists 
  a finite set 
  $\{\xi^{(1)}, \dots, \xi^{(J)}\} \subset S$
  with $J=J(\delta)$
  such that 
  $\sup_{\xi \in S}
   \min_{1 \le j \le J}
   \|\xi - \xi^{(j)}\|_{\infty} 
   < \delta/4
  $. 
  It follows that,
  for any $\delta>0$,
  \begin{eqnarray}
    \label{eq:E1-bb0}
    \Pr  
    \Big \{
    \min_{1 \le j \le J}
    \|\xi_{n} - \xi^{(j)}\|_{\infty} 
    \ge \delta/2
    \Big \}
    \le 
    \Pr
    \big \{
    \xi_{n} \not \in S^{\delta/4}
    \big \}
    \le \epsilon / 2,
  \end{eqnarray}
  for a sufficiently large $n$.

  In the view of 
  the compactness of $\mathcal{V}$, 
  for every $\eta>0$,
  there is a finite set
  $\{v_{1}, \dots, v_{K} \} \subset \mathcal{V}$
  with $K=K(\eta)$
  such that 
  $\sup_{v \in \mathcal{V}}
   \min_{1 \le k \le K}
   |v - v_{k}|
   < \eta 
  $.  
  Define the map 
  $\Pi_{\delta}:
  \mathcal{V} 
  \mapsto
  \{v_{k}\}_{k=1}^{K}$
  such that
  $|v - \Pi_{\eta}(v)| \le \eta$
  for every $v \in \mathcal{V}$.
  Theorem 1.5.7 of \cite{VW1996} with (\ref{eq:w-limit-1})
  implies that 
  for any $\delta>0$,
  there exists $\eta>0$ such that,
  for a sufficiently large $n$,
  \begin{eqnarray}
    \label{eq:E2-bb0}
    \Pr
    \big \{
    \|\xi_{n} - \xi_{n} \circ \Pi_{\eta}\|_{\infty}
    > \delta/2
    \big \}
    <
    \epsilon/2.
  \end{eqnarray}
  It follows from 
  (\ref{eq:E1-bb0})
  and 
  (\ref{eq:E2-bb0})
  that, for a sufficiently large $n$,
  \begin{eqnarray*}
    \Pr
    \Big \{
     \min_{1 \le j \le J}
     \|\xi_{n} - \xi^{(j)} \circ \Pi_{\eta}\|_{\infty}
     > \delta
    \Big \}
    <
    \epsilon,
  \end{eqnarray*}
  which yields 
  that,
  given a set 
  $\mathcal{M}_{j,\eta}(\delta)
   :=\{\xi \in \ell^{\infty}(\mathcal{V}): \|\xi - \xi^{(j)} \circ \Pi_{\eta}\|_{\infty} \le \delta\}
  $, we have
  \begin{eqnarray*}
    \Pr 
    \Big \{
    \big | 
    \bar{\nu}_{n}
    (\Xi_{n})
    {-}
    \bar{\nu}_{n}
    (\Xi_{0})
    \big |
    (w)
    \ge \epsilon
    \Big \}
    \le 
    \sum_{j=1}^{J}
    \Pr
    \Big \{
    \sup_{\xi \in \mathcal{M}_{j,\eta}(\delta)}
    \big | 
    \bar{\nu}_{n}
    (\Xi_{0} {+} n^{-1/2} \xi)
    {-}
    \bar{\nu}_{n}
    (\Xi_{0})
    \big |
    (w)
    \ge \epsilon
    \Big \}
    +
    \epsilon.
  \end{eqnarray*}
  Since $J$ is finite, 
  it suffices to show that, for each $j = 1, \dots, J$,
  \begin{eqnarray*}
    \sup_{\xi \in \mathcal{M}_{j,\eta}(\delta)}
    \big | 
    \bar{\nu}_{n}
    (\Xi_{0} {+} n^{-1/2} \xi)
    -
    \bar{\nu}_{n}
    (\Xi_{0})
    \big |
    (w)
    = o_p(1).
  \end{eqnarray*}

  An application of the triangle inequality yields 
  that, 
  for every
  $\xi \in \mathcal{M}_{j,\eta}(\delta)$,
  \begin{eqnarray*}
    \big | 
    \bar{\nu}_{n}
    (\Xi_{0} {+} n^{-1/2} \xi)
    -
    \bar{\nu}_{n}
    (\Xi_{0})
    \big | 
    (w)
    &\le& 
    \big |
    \bar{\nu}_{n}
    (\Xi_{0} {+} n^{-1/2} \xi)
    -
    \bar{\nu}_{n}
    (\Xi_{0} {+} n^{-1/2} \xi^{(j)})
    \big | 
    (w)
    \\
    &&+
    \big |
    \bar{\nu}_{n}
    (\Xi_{0} {+} n^{-1/2} \xi^{(j)})
    -
    \bar{\nu}_{n}
    (\Xi_{0})
    \big |
    (w).
  \end{eqnarray*}
  We separately consider two terms on the right-hand side of the above inequality. 
  First, 
  we can form an envelop function
  $
    I_{i,j}^{(1)}(\delta):=
    1 
    \{ 
    \max_{1 \le k \le K}
    |
      \tilde{V}_{i} + n^{-1/2} \xi^{(j)}
      \circ \Pi_{\eta}(V_{2i})
      - w
    | \le 
    n^{-1/2} \delta
    \}
  $
  for a collection of functions 
  \begin{eqnarray*}
    \Big \{
    1 
    \big \{ 
    \tilde{V}_{i}
    + 
    n^{-1/2}
    \xi(V_{2i})
    \le w
    \big \}
    -
    1 
    \big \{ 
    \tilde{V}_{i}
    + 
    n^{-1/2}
    \xi^{(j)} \circ \Pi_{\eta}(V_{2i})
    \le w
    \big \}
    : 
    \xi \in \mathcal{M}_{j,\eta}(\delta)
    \Big \}.
  \end{eqnarray*}
  We can write 
  \begin{eqnarray}
    \label{eq:delta1}
    \sup_{\xi \in \mathcal{M}_{j,\eta}(\delta)}
    \big |
    \bar{\nu}_{n}
    (\Xi_{0} {+} n^{-1/2} \xi)
    -
    \bar{\nu}_{n}
    (\Xi_{0} {+} n^{-1/2} \xi^{(j)})
    \big |
    (w)
    \le 
    n^{-1/2}
    \sum_{i=1}^{n}
    I_{i,j}^{(1)}(\delta)
    +
    \sqrt{n}
    E[ I_{1,j}^{(1)}(\delta) ].
  \end{eqnarray}
  Let 
  $\xi_{k}^{(j)}:= \xi^{(j)} \circ \Pi_{\eta}(v_{k})$
  for
  $k=1, \dots, K$.
  The second term on the right-hand of (\ref{eq:delta1})
  become arbitrarily small for a sufficiently small $\delta>0$,
  because we have
  \begin{eqnarray}
    \label{eq:E-1}
    \sqrt{n}
    E[ I_{1,j}^{(1)}(\delta)]
    \le 
    \sqrt{n}
    \sum_{k=1}^{K}
    \int_{-n^{-1/2} \delta}^{n^{-1/2} \delta}
    f_{\tilde{V}}
    (s - \xi_{k}^{(j)} + w)
    ds
    \le 
    \delta C_{1},
  \end{eqnarray}
  for some constant $C_{1}$.
  Applying the Markov inequality for 
  the first term on the right-hand of (\ref{eq:delta1}), 
  we obtain
  \begin{eqnarray*}
    \Pr 
    \Big \{
    n^{-1/2}
    \sum_{i=1}^{n}
    I_{i,j}^{(1)}(\delta)
    \ge \epsilon
    \Big \}
    \le 
    \epsilon^{-1}
    \sqrt{n}
    E[ I_{1,j}(\delta)],
  \end{eqnarray*}
  where 
  the right-hand side 
  becomes arbitrarily small 
  for a sufficiently small $\delta$
  due to (\ref{eq:E-1}).
  Thus,
  the right-hand side of (\ref{eq:delta1})
  converges to 0 in probability
  for a sufficiently small $\delta>0$.

  Next, we have the remaining term. 
  We have, 
  for every $i = 1, \dots, n$,
  \begin{eqnarray*}
    \big |
    1 
    \big \{ 
    \tilde{V}_{i}
    + 
    n^{-1/2}
    \xi^{(j)} \circ \Pi_{\eta}(V_{2i})
    \le w
    \big \}
    -
    1 
    \big \{ 
    \tilde{V}_{i}
    \le w
    \big \}
    \big |
    \le 
    I_{i,j}^{(2)}(\delta)
  \end{eqnarray*}
  where 
  $I_{i,j}^{(2)}(\delta)
    :=
    1 
    \big \{ 
    |
      \tilde{V}_{i} 
      - w
    | \le 
    n^{-1/2} 
    \max_{1 \le k \le K}
    |\xi_{k}^{(j)}|
    \big \}
   $.
   Using the Markov inequality, we can show that 
  \begin{eqnarray*}
    \Pr 
    \Big \{
    \big | 
    \bar{\nu}_{n}
    (\Xi_{0} {+} n^{-1/2} \xi^{(j)})
    -
    \bar{\nu}_{n}
    (\Xi_{0})
    \big |
    (w)
    \ge
    \epsilon
    \Big \}
    \le
    \epsilon^{-2}
    E
    [I_{i,j}^{(2)}(\delta)],
  \end{eqnarray*}
  where
  the right-hand side goes to zero 
  as $n\to \infty$
  because 
  $
    E
    \big[
     I_{i,j}^{(2)}(\delta)
    \big ]
    \le 
    C_{2} 
    n^{-1/2}
  $
  for some constant $C_{2}$.
  Hence 
  the proof is completed. 
\end{proof}


\vspace{0.5cm}

\begin{lemma}
\label{lemma:l-approx}  
Let $\phi_{n} \in \ell({\cal W})$ be the function defined in (\ref{eq:def-phi}).
Suppose that  $\sqrt{n}(F_{n} - F_{0} )$  converges
in distribution to a tight, random element in  $\big (\ell^{\infty}(\mathcal{V}) \big)^{2}$.
Then,   
\begin{eqnarray*}
\sqrt{n} \big ( \phi_{n}( F_{n}) - \phi(F_{0}) \big ) = \nu_{n}(F_{0}) +
\phi_{F_{0}}^{\prime }\big ( \sqrt{n} (F_{n} - F_{0}) \big ) + o_p(1),
\end{eqnarray*}
uniformly in ${\cal W}$,
where  $\phi_{F_{0}}^{\prime }$  is  the Hadamard derivative  given in 
Lemma \ref{lemma:H-diff-2}.
\end{lemma}

\begin{proof}[\textbf{Proof}]
  By definition, we can write 
  \begin{eqnarray*}
    \sqrt{n}
    \big (
    \phi_{n}
    (F_{n})
    -
    \phi
    (F_{0})
    \big )
    =  
    \nu_{n}
    (F_{n})
    + 
    \sqrt{n}
    \big (
    \phi
    (F_{n})
    -
    \phi
    (F_{0})
    \big ).
  \end{eqnarray*}
  First, Lemma \ref{lemma:zero} shows that,
  uniformly in $\mathcal{W}$,
  \begin{eqnarray*}
    \nu_{n}
    (F_{n})
    =
    \nu_{n}
    (F_{0})
    +
    o_p(1).
  \end{eqnarray*}
  Since the map $\phi$
  is Hadamard differentiable,
  the functional delta method 
  in Theorem 3.9.4 of \cite{VW1996} 
  with 
  Lemma \ref{lemma:H-diff-2} implies that 
  \begin{eqnarray*}
    \sqrt{n}
    \big (
    \phi
    (F_{n})
    -
    \phi
    (F_{0})
    \big )
    =
    \phi_{F_{0}}'
    \big (
      \sqrt{n}
      (F_{n} - F_{0})
    \big ) 
    +
    o_p(1).
  \end{eqnarray*}
  Hence the desired result follows.
\end{proof}




\vspace{0.5cm} 
\begin{proof}
  [\textbf{Proof of Lemma \ref{lemma:asym-basic}}]
  The result follows from the functional central limit theorem
  for empirical distribution functions.
  See Chapter 2 of \cite{VW1996} for instance.
\end{proof}


\vspace{0.5cm} 
\begin{proof}
  [\textbf{Proof of Proposition \ref{proposition:F-convergence}}]
  Let $x \in \mathcal{X}$ be fixed.
  From its definition in (\ref{eq:est-main}), 
  the counterfactual distribution 
  estimator 
  $\hat{F}_{Y_{t}(0)|X=x, D_{t}=1}$
  can be considered as an empirical distribution 
  indexed by estimated distribution functions 
  $\hat{F}_{Y_{t-1|X=x, D_{t}=0}}$
  and 
  $\hat{F}_{Y_{t-1|X=x, D_{t}=1}}$. 
  Thus, there exists some map 
  $\phi_{n}: 
   \mathcal{C}(\mathcal{Y}_{t-1|x,0})
   \times
   \mathcal{C}(\mathcal{Y}_{t-1|x,0})
   \mapsto
   \mathcal{C}(\mathcal{Y}_{t-1|x,1}(0))
  $
  such that 
  \begin{eqnarray*}
    \hat{F}_{Y_{t}(0)|X=x, D_{t}=1}
    = 
    \phi_{n}
    (
    \hat{F}_{Y_{t-1|X=x, D_{t}=0}},
    \hat{F}_{Y_{t-1|X=x, D_{t}=1}} 
    ).
  \end{eqnarray*}
  The empirical processes 
  of 
  $\hat{F}_{Y_{t-1|X=x, D_{t}=0}}$
  and 
  $\hat{F}_{Y_{t-1|X=x, D_{t}=1}}$
  are
  given by 
  $\hat{G}_{t-1, x}^{(0)}$
  and 
  $\hat{G}_{t-1, x}^{(1)}$, respectively,
  and 
  Lemma \ref{lemma:asym-basic}  
  shows that they jointly have 
  a tight limit asymptotically. 
  Also, notice that 
  $
  \phi_{n}
  (
  F_{Y_{t-1|X=x, D_{t}=0}},
  F_{Y_{t-1|X=x, D_{t}=1}} 
  )
  =
  \tilde{F}_{Y_{t}(0)|X=x, D_{t}=1}
  $
  and 
  that
  the empirical process 
  $\tilde{F}_{Y_{t}(0)|X=x, D_{t}=1}$
  is 
  $\tilde{G}_{t, x}^{(0)}$
  as defined in (\ref{eq:ep-tilde}). 
  It follows from 
  Lemma \ref{lemma:H-diff-2}
  and \ref{lemma:l-approx}
  that 
  \begin{eqnarray*}
    \sqrt{n}
    \big (
    \hat{F}_{Y_{t}(0)|X=x, D_{t}=1}
    - 
    F_{Y_{t}(0)|X=x, D_{t}=1}
    \big )
    = 
    r_{x}^{(0)}
    \tilde{G}_{t, x}^{(0)}
    + 
    \kappa_{x}
    \big (
    \hat{G}_{t-1, x}^{(0)},
    \hat{G}_{t-1, x}^{(1)}
    \big )
    +
    o_p(1),
  \end{eqnarray*}
  uniformly in $\mathcal{Y}_{t|x, 1}(0)$,
  where
  $\kappa_{x}$
  represents 
  the Hadamard derivative of 
  $\phi_{n}$
  as shown in Lemma \ref{lemma:H-diff-2}.
  Hence 
  the extended continuous mapping theorem with 
  Lemma \ref{lemma:asym-basic} yields 
  the desired result.
\end{proof}


\vspace{0.5cm} 
\begin{proof}
  [\textbf{Proof of Theorem \ref{theorem:Q-convergence}}]
  Let $x \in \mathcal{X}$ be fixed.
  When  
  $\hat{F}_{Y_{t}(j)|X=x, D_{t}=1}(y)$
  is  
  weakly increasing in $y$, 
  we can show that 
  the corresponding quantile function 
  $\hat{F}_{Y_{t}(j)|X=x, D_{t}=1}^{-1}(\tau)$
  is Hadamard differentiable. 
  It follows from the functional delta method
  that 
  \begin{eqnarray*}
    \sqrt{n}
    \big (
    \hat{F}_{Y_{t}(j)|X=x, D_{t}=1}^{-1}(\tau)
    -
    F_{Y_{t}(j)|X=x, D_{t}=1}^{-1}(\tau)
    \big )
    \rightsquigarrow 
    \Big(
    \frac{
      \mathbb{Z}_{x}^{(j)}
    }{
      f_{Y_{t}(j)|X=x, D_{t}=1}
    }
    \Big )
    \circ 
    F_{Y_{t}(j)|X=x, D_{t}=1}^{-1}(\tau),
  \end{eqnarray*}
  as a stochastic process indexed by 
  $\tau \in \mathcal{T}$ and $j \in \{0,1\}$.
  Hence the desired result holds.
\end{proof} 
\vspace{0.5cm} 


\vspace{0.5cm} 
\begin{proof}
  [\textbf{Proof of Corollary \ref{corollary:ks}}]
The result follows from the continuous mapping theorem.  See, for example, Section 2.1 of \cite{kosorok-2007}.
\end{proof}
\vspace{0.5cm}


We now prove a technical lemma, which is a bootstrap version of Lemma \ref{lemma:l-approx}. 
We respectively denote the bootstrap counterpart 
of $\phi_{n}$ and $F_{n}$ by $\phi_{n}^{\ast}$ and $F_{n}^{\ast}$, which are obtained 
through bootstrap with some random weights $(W_{1}, \dots, W_{n})$ satisfying Assumption B.



\vspace{0.5cm}

\begin{lemma}
\label{lemma:l-approx-B}

Suppose that 
both 
$\sqrt{n}(F_{n} - F_{0})$
and  
$\sqrt{n}(F_{n}^{\ast} - F_{0})$  converge in distribution to  
some tight random elements unconditional on the original sample.  Then,  
uniformly in ${\cal W}$,
\begin{eqnarray*}
\sqrt{n} \big ( \phi_{n}^{\ast}(F_{n}^{\ast}) - \phi_{n}(F_{n}) \big ) = 
\sqrt{n} \big ( \phi_{n}^{\ast}(F_{0}) - \phi_{n}(F_{0}) \big ) +
\phi_{F_{0}}^{\prime }\big ( \sqrt{n} (F_{n}^{\ast} - F_{n} ) \big ) +
o_p(1).
\end{eqnarray*}
\end{lemma}

\begin{proof}
  Let 
  $\tilde{\nu}_{n}^{\ast}(F)
   := 
   \sqrt{n}
   \big(
    \phi_{n}^{\ast}(F)
    -
    \phi(F)
   \big)
  $
  for 
  $F \in (\ell^{\infty}(\mathcal{V}))^2$. 
  We have 
  \begin{eqnarray*}
    \sqrt{n}
    \big (
    \phi_{n}^{\ast}(F_{n}^{\ast}) 
    -
    \phi(F_{0}) 
    \big )
    =  
    \tilde{\nu}_{n}^{\ast}(F_{n})
    +
    \sqrt{n}
    \big (
    \phi(F_{n}^{\ast})
    -
    \phi(F_{0})
    \big ).
  \end{eqnarray*}
  By a similar argument used to prove  
  Lemma \ref{lemma:zero}, we can show that,
  uniformly in $\mathcal{W}$,
  \begin{eqnarray}
    \label{eq:tt1}
    \tilde{\nu}_{n}^{\ast}(F_{n})
    =     
    \tilde{\nu}_{n}^{\ast}(F_{0})
    +
    o_p(1) .   
  \end{eqnarray}
  Also
  $\phi$ is the Hadamard differentiable function $\phi$
  and 
  $
    \sqrt{n}(
    F_{n}^{\ast}
    -
    F_{0}
    )
  $
  converges in distribution to 
  a tight random element unconditional on the original sample.
  Thus the functional delta method implies that,
  uniformly in $\mathcal{W}$,
  \begin{eqnarray}
    \label{eq:tt2}
    \sqrt{n}
    \big (
    \phi(F_{n}^{\ast})
    -
    \phi(F_{0})
    \big )
    = 
    \phi_{F_{0}}'
    \big (
    \sqrt{n}
    (F_{n}^{\ast}
    -
    F_{0}
    )
    \big )
    +o_p(1).
  \end{eqnarray}
  It follows from (\ref{eq:tt1}) and (\ref{eq:tt2}) that,
  uniformly in $\mathcal{W}$,
  \begin{eqnarray*}
    \sqrt{n}
    \big (
    \phi_{n}^{\ast}(F_{n}^{\ast}) 
    -
    \phi(F_{0}) 
    \big )
    =  
    \tilde{\nu}_{n}^{\ast}(F_{0})
    +
    \phi_{F_{0}}'
    \big (
    \sqrt{n}
    (F_{n}^{\ast}
    -
    F_{0}
    )
    \big )
    +
    o_p(1),
  \end{eqnarray*}
  which together with Lemma \ref{lemma:l-approx} yields the desired result
  because 
  $\phi_{F_{0}}'$
  is a linear map.
\end{proof}


\vspace{0.5cm} 
\begin{proof}
  [\textbf{Proof of Lemma \ref{lemma:asym-basic-B}}]
  The result follows from Theorem 3.6.13 of \cite{VW1996}.
  Thus we omit the detail.
\end{proof}


\vspace{0.5cm} 
\begin{proof}
  [\textbf{Proof of Theorem \ref{theorem:F-convergence-B}}]
  Let $x \in \mathcal{X}$ be fixed.
  First we wish to show that 
  $
    \hat{Z}_{x}^{\ast}
    \rightsquigarrow^{p} 
    \mathbb{Z}_{x}
  $,
  where 
  $
    \hat{Z}_{x}^{\ast}
    :=
    (\hat{Z}_{x}^{(0) \ast}, \hat{Z}_{x}^{(1) \ast})'
  $
  and 
  $
    \mathbb{Z}_{x}
    :=
    (\mathbb{Z}_{x}^{(0)},\mathbb{Z}_{x}^{(1)})'        
  $.
  Define 
  $\tilde{Z}_{x}^{\ast}:=
   (
    \tilde{Z}_{x}^{(0)\ast},
   \hat{Z}_{x}^{(1) \ast} 
   )'
  $,
  where
  $
    \tilde{Z}_{x}^{(0)\ast}
    :=
    r_{x}^{(0)}
    \tilde{G}_{t, x}^{(0) \ast}
    +
    \kappa_{x}
    (
    \hat{G}_{t-1, x}^{\ast},
    \hat{G}_{t, x}^{\ast}
    ),
  $
  with 
  $\tilde{G}_{t, x}^{(0) \ast}$,
  $\hat{G}_{t-1, x}^{\ast}$
  and 
  $\hat{G}_{t, x}^{\ast}$
  denoting the bootstrap version 
  of the empirical processes 
  $\tilde{G}_{t, x}^{(0)} $,
  $\hat{G}_{t-1, x}$
  and 
  $\hat{G}_{t, x}$,
  respectively.
  By the triangle inequality, we obtain
  \begin{eqnarray}
    \label{eq:E1}
    \sup_{h \in BL_{1}}
    \big |
    E_{M}
    [ 
    h( \hat{Z}_{x}^{\ast} )
    ]
    - 
    E
    [ 
    h( \mathbb{Z}_{x} )
    ]
    \big | 
    &\le& 
    \sup_{h \in BL_{1}}
    \big |
    E_{M}[ h ( \hat{Z}_{x}^{\ast}  )  ]
    - 
    E_{M}[ h ( \tilde{Z}_{x}^{\ast} ) ]
    \big | \\
    && \label{eq:E2}
    +
    \sup_{h \in BL_{1}}
    \big |
    E_{M}[ 
    h (\tilde{Z}_{x}^{\ast} )
    ]
    - 
    E
    [h( \mathbb{Z}_{x})]
    \big |. 
  \end{eqnarray}
  It suffices to show that 
  (\ref{eq:E1})
  and 
  (\ref{eq:E2})
  converge in probability to zero, separately.

  We consider 
  (\ref{eq:E1}).
  Because 
  $ 
    \big |
    E_{M}[ h ( \hat{Z}_{x}^{\ast}  )  ]
    - 
    E_{M}[ h ( \tilde{Z}_{x}^{\ast} ) ]
    \big | 
    \le 
    E_{M}
    \big |
    h ( \hat{Z}_{x}^{\ast}  )  
    - 
    h ( \tilde{Z}_{x}^{\ast} ) 
    \big | 
  $, 
  we have 
  \begin{eqnarray}
    \label{eq:bd-000}
    E
    \Big [
    \sup_{h \in BL_{1}}
    \big |
    E_{M}[ h ( \hat{Z}_{x}^{\ast}  )  ]
    - 
    E_{M}[ h ( \tilde{Z}_{x}^{\ast} ) ]
    \big | 
    \Big ]
    \le 
    E
    \Big [
    \sup_{h \in BL_{1}}
    \big |
     h ( \hat{Z}_{x}^{\ast}  )  
    - 
     h ( \tilde{Z}_{x}^{\ast} )
    \big | 
    \Big ].
  \end{eqnarray}
  Let $\epsilon>0$ be fixed
  and 
  define 
  $I_{n,\epsilon}^{\ast}
   :=
   1\{ \|\hat{Z}_{x}^{\ast} - \tilde{Z}_{x}^{\ast}\|_{\infty} > \epsilon\}
  $.
  Lemma \ref{lemma:asym-basic-B} and \ref{lemma:l-approx-B} 
  imply 
  that 
  $\lim_{n \to \infty}E[I_{n,\epsilon}^{\ast}] \le \epsilon$,
  while 
  $ \sup_{h \in BL_{1}}
    \big |
     h ( \hat{Z}_{x}^{\ast}  )  
    - 
     h ( \tilde{Z}_{x}^{\ast} )
    \big |  \le 2
  $.
  It follows that 
  \begin{eqnarray}
    \label{eq:bd-001}
    E
    \Big [
    \sup_{h \in BL_{1}}
    \big |
     h ( \hat{Z}_{x}^{\ast}  )  
    - 
     h ( \tilde{Z}_{x}^{\ast} )
    \big | 
    \cdot
    I_{n,\epsilon}^{\ast}
    \Big ]
    \le 2 \epsilon.
  \end{eqnarray}
  Also we can show that 
  \begin{eqnarray}
    \label{eq:bd-002}
    E
    \Big [
    \sup_{h \in BL_{1}}
    \big |
     h ( \hat{Z}_{x}^{\ast}  )  
    - 
     h ( \tilde{Z}_{x}^{\ast} )
    \big | 
    \cdot
    (1-I_{n,\epsilon}^{\ast})
    \Big ]
    \le \epsilon,
  \end{eqnarray}
  because 
  $
  \sup_{h \in BL_{1}}
  \big |
  h ( \hat{Z}_{x}^{\ast}  )  
  - 
  h ( \tilde{Z}_{x}^{\ast} )
  \big | 
  \le 
  \|\hat{Z}_{x}^{\ast}(y) - \tilde{Z}_{x}^{\ast}(y)\|_{\infty}
  $.
  It follows from 
  (\ref{eq:bd-001})
  and 
  (\ref{eq:bd-002})
  that 
  the right-hand side of 
  (\ref{eq:bd-000})
  is bounded by $3\epsilon$.
  Since 
  $\epsilon$ is arbitrary, 
  an application of the Markov inequality yields the convergence 
  of 
  (\ref{eq:E1})
  to 0 in probability.

  Consider (\ref{eq:E2}).
  Using Lemma \ref{lemma:asym-basic-B} together 
  with
  the continuous mapping theorem, 
  we can show that 
  (\ref{eq:E2}) converges to 0 in probability.
  Hence we obtain the desired result.

  We now consider validity of exchangeable bootstrap for the CQTT.
  Theorem 3.9.11 of \cite{VW1996}
  shows that 
  the functional delta method 
  can apply for
  Hadamard differentiable maps 
  under 
  resampling. 
  Since the map from distribution to quantile 
  is Hadamard differentiable, the desired result follows. 
\end{proof}


\pagebreak

\section*{Tables and Figures}

\newcolumntype{.}{D{.}{.}{-1}} 
\ctable[caption={Monte Carlo Simulations, DGP 1},label=,pos=!htbp,]{lcccccccc}{\tnote[]{\textit{Notes}: Each Monte Carlo simulation uses 1000 bootstrap iterations.  Each cell lists the bias and the rejection probabilities from 1000 Monte Carlo simulations. To calculate empirical rejection probabilities, we set the nominal size to be 5\%.  }}{\FL
\multicolumn{2}{l}{\bfseries }&\multicolumn{3}{c}{\bfseries DDID}&\multicolumn{1}{c}{\bfseries }&\multicolumn{3}{c}{\bfseries CIC}\NN[5pt]
&\multicolumn{1}{c}{$N$}&\multicolumn{1}{c}{0.1}&\multicolumn{1}{c}{0.5}&\multicolumn{1}{c}{0.9}&\multicolumn{1}{c}{}&\multicolumn{1}{c}{0.1}&\multicolumn{1}{c}{0.5}&\multicolumn{1}{c}{0.9}\ML
{\bfseries TE=0}&&&&&&&\NN[5pt]
\hspace{5pt}\textit{Bias}&100 &0.044&0.045&0.081&&0.012&-0.097&-0.295\NN[5pt]
&200 &0.016&0.021&0.066&&-0.009&-0.048&-0.141\NN[5pt]
&500 &0.016&0.008&0.023&&-0.005&-0.042&-0.074\NN[5pt] 
\hspace{5pt}\textit{Rej. Prob.} &100      &0.042&0.037&0.023&&0.042&0.041&0.100\NN[5pt]
&200 &0.049&0.050&0.044&&0.051&0.056&0.076\NN[5pt]
&500 &0.043&0.047&0.034&&0.035&0.043&0.069\ML
{\bfseries TE=1}&&&&&&&&\NN[5pt]
\hspace{5pt} \textit{Bias}&100 &0.059&0.064&0.109&&0.251&-0.051&-0.293\NN[5pt]
&200 &0.031&0.027&0.049&&0.128&-0.052&-0.200\NN[5pt]
&500 &0.014&0.019&0.025&&0.053&-0.019&-0.090\NN[5pt]
\hspace{5pt} \textit{Rej. Prob.}&100 &0.397&0.675&0.359&&0.409&0.617&0.548\NN[5pt]
&200 &0.742&0.949&0.703&&0.713&0.859&0.614\NN[5pt]
&500 &0.994&1.000&0.992&&0.983&0.992&0.756\LL
}

\newcolumntype{.}{D{.}{.}{-1}} 
\ctable[caption={Monte Carlo Simulations, DGP 2},label=,pos=!tbp,]{lccccccc}{\tnote[]{\textit{Notes}: Each Monte Carlo simulation uses 1000 bootstrap iterations.  Each cell lists the bias and the root mean squared error from 1000 Monte Carlo simulations.  Here, $\bar{\rho}$  controls whether or not the the Copula Invariance assumption is violated.  When $\bar{\rho} = 0$, the Copula Invariance assumption holds; the further away  $\bar{\rho}$ is from 0, the more strongly the Copula Invariance assumption is violated.}}{\FL
\multicolumn{1}{l}{\bfseries }&\multicolumn{3}{c}{\bfseries TE=0}&\multicolumn{1}{c}{\bfseries }&\multicolumn{3}{c}{\bfseries TE=1}\NN[5pt]
\multicolumn{1}{c}{$\bar{\rho}$}&\multicolumn{1}{c}{0.1}&\multicolumn{1}{c}{0.5}&\multicolumn{1}{c}{0.9}&\multicolumn{1}{c}{}&\multicolumn{1}{c}{0.1}&\multicolumn{1}{c}{0.5}&\multicolumn{1}{c}{0.9}\ML
\textit{Bias} &&&&&&&\NN[5pt]
~~0.00 &0.020&0.034&0.037&&0.023&0.023&0.050\NN[5pt]
~~0.05 &0.073&0.023&0.012&&0.088&0.029&0.008\NN[5pt]
~~0.10 &0.121&0.028&-0.033&&0.112&0.019&-0.032\NN[5pt]
~~0.50 &0.425&0.013&-0.374&&0.435&0.027&-0.353\NN[10pt]
\textit{RMSE} &&&&&&&\NN[5pt]
~~0.00 &0.348&0.261&0.340&&0.342&0.248&0.359\NN[5pt]
~~0.05 &0.348&0.256&0.324&&0.358&0.258&0.342\NN[5pt]
~~0.10 &0.374&0.260&0.352&&0.374&0.259&0.346\NN[5pt]
~~0.50 &0.565&0.272&0.529&&0.566&0.264&0.508\LL
}


\clearpage

\newcolumntype{.}{D{.}{.}{-1}} 
\ctable[caption={Summary Statistics (averages)},label=,pos=!tbp,]{lrrrr}{\tnote[]{\textit{Notes}:
The second and third columns report sample averages for states that raised their minimum wage during the first quarter of 2007 (treated states) and states that had their minimum wage equal to the federal minimum wage for the entire period.  
The last column presents differences between the figures in the second and third columns 
with $p$-values in parentheses.

\textit{Sources}: Panel data from the Current Population Survey (CPS) \citep{ipums-cps-2015}. \\ \ \vspace{2cm}}}{\FL
\multicolumn{1}{l}{}&\multicolumn{1}{c}{Treated States}&\multicolumn{1}{c}{Untreated States}&\multicolumn{1}{c}{Difference (p-value)}\ML
White&\multicolumn{1}{c}{0.89}&\multicolumn{1}{c}{0.88}&\multicolumn{1}{c}{0.014 (0.06)} \NN[5pt]
Male&\multicolumn{1}{c}{0.48}&\multicolumn{1}{c}{0.49}&\multicolumn{1}{c}{-0.008 (0.52)}\NN[5pt]
College Degree&\multicolumn{1}{c}{0.45}&\multicolumn{1}{c}{0.41}&\multicolumn{1}{c}{0.039 (0.00)}\NN[5pt]
Log Earnings&\multicolumn{1}{c}{6.39}&\multicolumn{1}{c}{6.33}&\multicolumn{1}{c}{0.065 (0.00)}\LL
}

\newcolumntype{.}{D{.}{.}{-1}} 
\ctable[caption={Conditional QTT Estimates},label=Subgroup,pos=!tbp,]{lllccccc}{\tnote[]{\textit{Notes:}
Conditional QTTs by race, gender, and education subgroups.  $N$ is the number of observations in each group.  The column of ``Reject $H_{0}$'' reports whether the null of no effect at any quantile is rejected using the Kolmogorov-Smirnov test with nominal size of 5\% using equally spaced quantiles from 0.05 to 0.95 by 0.01 and is based on 1000 bootstrap iterations to calculate critical values.  The last three columns report conditional QTTs at the 0.1, 0.5, and 0.9 quantiles.  Standard errors are pointwise and computed using the bootstrap with 1000 iterations.
  
\textit{Sources:}  Panel data from the Current Population Survey (CPS) \citep{ipums-cps-2015}}}{\FL
\multicolumn{3}{c}{Subgroup}&&&\multicolumn{3}{c}{Quantile}\NN
\cline{1-3}\cline{6-8} \NN[5pt]
Race & Gender & Education &
$N$&
\multicolumn{1}{c}{Reject $H_0$}
& 0.1 & 0.5 & 0.9 \ML
White& Male& College&1617&&0.004&0.007&0.000\NN[5pt]
&&&&&(0.048)&(0.031)&(0.037)\NN[5pt]
& & Non-College&2306&&-0.021&-0.019&0.075\NN[5pt]
&&&&&(0.068)&(0.028)&(0.068)\NN[5pt]
& Female& College&1629&Yes&-0.029&0.027&-0.046\NN[5pt]
&&&&&(0.056)&(0.033)&(0.052)\NN[5pt]
& & Non-College&1980&&-0.038&-0.043&-0.086\NN[5pt]
&&&&&(0.054)&(0.035)&(0.068)\NN[5pt]
Non-White& Male& College&156&&-0.492&-0.07&0.292\NN[5pt]
&&&&&(0.265)&(0.161)&(0.268)\NN[5pt]
& & Non-College&282&Yes&-0.087&-0.044&0.036\NN[5pt]
&&&&&(0.186)&(0.095)&(0.136)\NN[5pt]
& Female& College&209&&0.097&0.033&0.007\NN[5pt]
&&&&&(0.172)&(0.088)&(0.121)\NN[5pt]
& & Non-College&340&Yes&-0.25&-0.026&-0.043\NN[5pt]
&&&&&(0.175)&(0.086)&(0.162)\LL
}


\clearpage

\begin{figure}[h!]
\center
\caption{Conditional QTTs of Minimum Wage Increase on Earnings}
  (with 95\% Confidence Bands)

\includegraphics{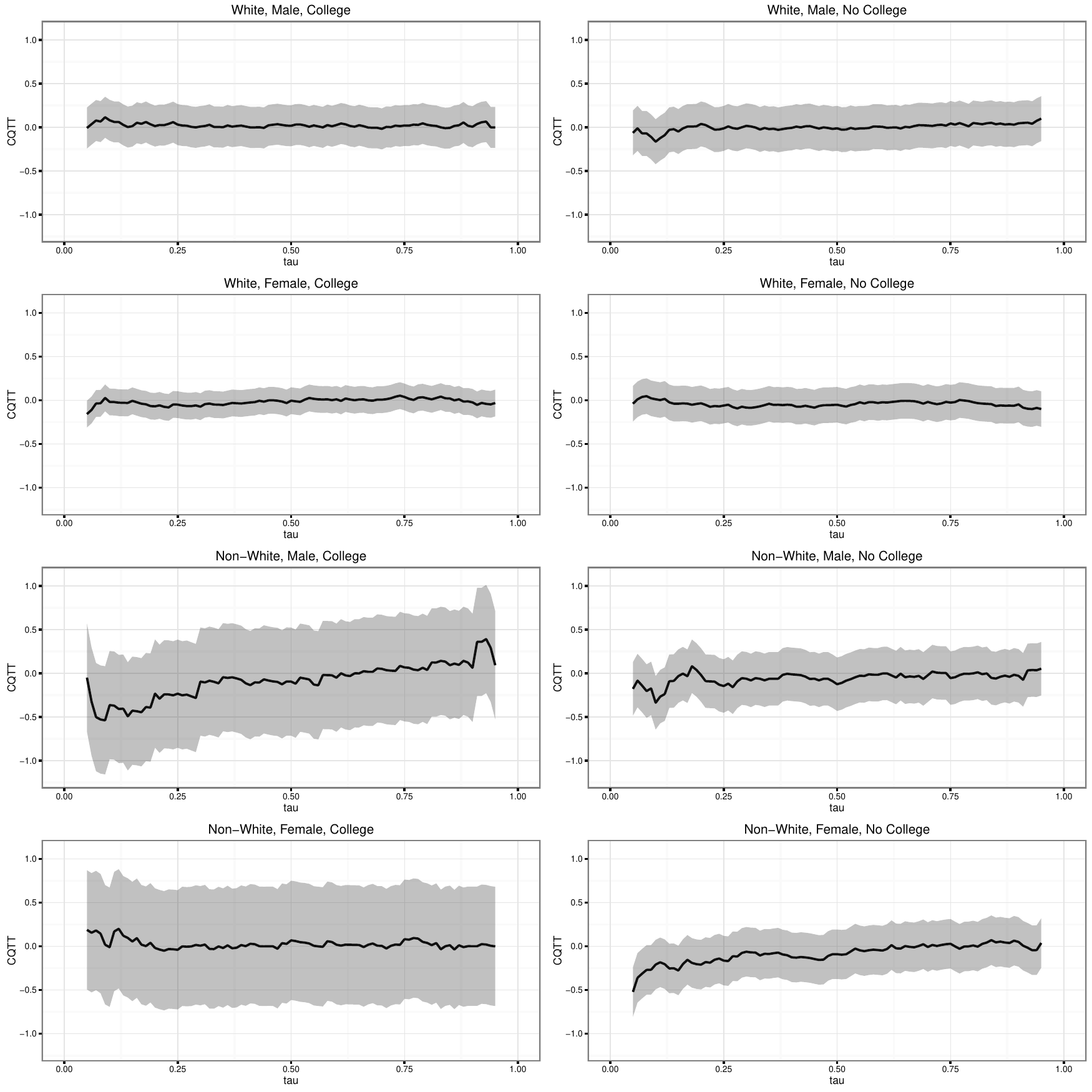}
\end{figure}

\begin{minipage}[c]{15cm}
\small
\textit{Notes:}
Conditional QTT estimates for groups formed by race, gender, and education.  The figure also provides 95\% confidence bands.  These are formed by inverting Kolmogorov-Smirnov statistics and based on 1000 bootstrap iterations.

\textit{Sources}: Panel data from the Current Population Survey (CPS) \citep{ipums-cps-2015}.
\end{minipage}


\end{document}